\documentclass[journal]{IEEEtran}

\ifCLASSINFOpdf

\else

\fi

\usepackage{amsthm}
\usepackage{epsfig} 
\usepackage{epstopdf}
\usepackage{amsmath}
\usepackage{amsthm}
\usepackage{comment}
\usepackage{url}
\usepackage{algorithm}
\usepackage{algorithmic}

\usepackage{amsfonts}
\usepackage{bm}
\usepackage{color}
\usepackage{cite}
\usepackage{stfloats}
\usepackage{chemarrow}
\usepackage{extarrows}
\allowdisplaybreaks

\newboolean{showcomments}
\setboolean{showcomments}{true}
\newcommand{\wang}[1]{\ifthenelse{\boolean{showcomments}}
	{ \textcolor[rgb]{1,0,1}{(ZW:  #1)}}{}}
\newcommand{\fliu}[1]{\ifthenelse{\boolean{showcomments}}
	{ \textcolor{blue}{(FL:  #1)}}{}}
\newcommand{\ychen}[1]{\ifthenelse{\boolean{showcomments}}
	{ \textcolor{green}{(ZP:  #1)}}{}}
\newcommand{\slow}[1]{\ifthenelse{\boolean{showcomments}}
	{ \textcolor{blue}{(SL:  #1)}}{}}

\theoremstyle{definition}
\newtheorem{theorem}{Theorem}
\newtheorem{lemma}[theorem]{Lemma}

\theoremstyle{definition}

\title{ Distributed Optimal Frequency Control Considering a Nonlinear Network-Preserving Model}

\begin{document}
%


%

\author{
	Zhaojian~Wang,
		Feng~Liu,
	John Z. F. Pang,
        Steven~Low,~\IEEEmembership{Fellow,~IEEE} and Shengwei~Mei,~\IEEEmembership{Fellow,~IEEE}
        \thanks{This work was supported  by the National Natural Science Foundation
				of China (No. 51677100, U1766206, No. 51621065), the US National Science Foundation through awards EPCN 1619352, CCF 1637598, CNS 1545096, ARPA-E award DE-AR0000699, and Skoltech through Collaboration Agreement 1075-MRA.  (Corresponding author: Shengwei Mei)   }     

		\thanks{Z. Wang, F. Liu and S. Mei are with the State Key Laboratory of Power System and  Department of Electrical Engineering, Tsinghua University, Beijing,
				China, 100084. e-mail: (meishengwei@mail.tsinghua.edu.cn).}
		\thanks{S. H. Low and J. Pang  are with the Department
				of Electrical Engineering, California Institute of Technology, Pasadena, CA, USA, 91105. e-mail:(slow@caltech.edu).}
}

\markboth{Journal of \LaTeX\ Class Files,~Vol.~xx, No.~xx, xx~xxxx}%
{Shell \MakeLowercase{\textit{et al.}}: Bare Demo of IEEEtran.cls for IEEE Journals}

\maketitle

\begin{abstract}
    This paper addresses the distributed optimal frequency control of power systems considering a network-preserving model with nonlinear power flows and excitation voltage dynamics. Salient features of the proposed distributed control strategy are fourfold: i) nonlinearity is considered to cope with large disturbances; ii) only a part of generators are controllable; iii) no load measurement is required; iv) communication connectivity is required only for the controllable generators. To this end, benefiting from the concept of ``virtual load demand", we first design the distributed controller for the controllable generators by leveraging the  primal-dual decomposition technique. We then propose a method to estimate the virtual load demand of each controllable generator based on local frequencies. We derive incremental passivity conditions for the uncontrollable generators. Finally, we prove that the closed-loop system is asymptotically stable and its equilibrium attains the optimal solution to the associated economic dispatch problem. Simulations, including small and large-disturbance scenarios, are carried on the New England system, demonstrating the effectiveness of our design.
\end{abstract}



\section{Introduction}

Frequency restoration and economic dispatch (ED) are two important problems in power system operation. Conventionally, they are implemented hierarchically in a centralized fashion, where  the former is addressed in a fast time scale while the latter in a slow time scale\cite{Cai:Distributed,Shiltz:Integration}. While this centralized hierarchy  works well for the traditional power system, it may not be able to keep pace with the fast development of our power system due to: 1) slow response, 2) insufficient flexibility, 3) low privacy, 4) intense communication, and 5) single point of failure issue. In this regard, the idea of  breaking such a hierarchy is proposed in \cite{jokic:real,Breaking:Dorfler}. 

In \cite{Li:Connecting}, an intrinsic connection between the asymptotic stability of the dynamical frequency control system with the ED problem is proposed. It leads to a so-called inverse-engineering methodology for designing optimal frequency controllers, where the (partial) primal-dual gradient algorithm plays an essential role \cite{Feijer:Stability,Cai:Distributed}. 
When designing optimal frequency controllers, in choice of power flow models, including the linear model (usually associated with DC power flow, e.g. \cite{Li:Connecting,Changhong:Design,Mallada-2017-OLC-TAC,Kasis:Primary1,zhao:aunified,Pang:Optimal,Wang:Decentralized,Zhang:Achieving,Distributed_I:Wang,Distributed_II:Wang}) and the nonlinear model (usually associated with AC power flow, e.g. \cite{Stegink:aport,stegink:stabilization,Stegink:aunifying, Zhang:A, Zhao:Optimal, Trip:optimal}), is crucial. While the closed-loop system can be interpreted in a linear model as carrying out a primal-dual algorithm for solving ED, this interpretation of frequency control may not
 hold in a nonlinear model. In addition to nonlinear power flow, excitation voltage dynamics are considered in \cite{Stegink:aport, stegink:stabilization,Stegink:aunifying}, making the model more realistic. 

The aforementioned idea is further developed to enable the design of distributed optimal frequency controllers. Roughly speaking, the works of distributed optimal frequency control can be divided into two categories in terms of different power system models:  network-reduced models e.g. \cite{Li:Connecting,Zhang:Achieving,Stegink:aport,Wang:Decentralized,Stegink:aunifying,Distributed_I:Wang,Distributed_II:Wang} and network-preserving models e.g. \cite{Changhong:Design,Mallada-2017-OLC-TAC,Kasis:Primary1,stegink:stabilization,zhao:aunified,Pang:Optimal}. In  network-reduced models, generators and/or loads are aggregated and treated as one bus or control area, which are connected to each other through tie lines. In \cite{Li:Connecting,Zhang:Achieving}, aggregated generators in each area are driven by automatic generation control (AGC) to restore system frequency. 
\cite{Stegink:aport,Wang:Decentralized,Stegink:aunifying,Distributed_I:Wang,Distributed_II:Wang} further consider  both the aggregated generators and load demands in  frequency control. 
In network-preserving models, generator and load buses are separately handled with different dynamic models and coupled by power flows, rendering a set of differential algebraic equations (DAEs). In \cite{Changhong:Design}, an optimal load control (OLC) problem is formulated and a primary load-side control is derived as a partial primal-dual gradient algorithm for solving the OLC problem. The design approach is extended to secondary frequency control (SFC) that restores  nominal frequency in \cite{Mallada-2017-OLC-TAC}.  It is generalized in \cite{Kasis:Primary1}, where passivity condition guaranteeing stability is proposed for each local bus. Then, a unified framework combining load and generator control is advocated in \cite{zhao:aunified}. A similar model is also utilized in \cite{Pang:Optimal}, where only limited control coverage is needed. Similar to \cite{Stegink:aunifying}, the Hamiltonian method is  used to analyze the network-preserving model in \cite{stegink:stabilization}. Compared with the  network-reduced model, the network-preserving model describes power systems more precisely and appear more suitable for analyzing interactions between different control areas. Therefore, we specifically consider the network-preserving model in this work.  

Almost all of aforementioned works assume that all buses are controllable and load demands at all buses are accurately measurable, especially for those proposing secondary frequency control. Moreover, it is usually assumed that the communication network has the same topology of the power grid. These assumptions are strong and arguably unrealistic for practice. First,  only a part of generators and controllable load buses can participate in frequency control in practice. Second, the communication network may not be identical to the topology of the power grid. Third, it is difficult to accurately measure the load injection on every bus. In extreme cases, even the number of load buses is unknown in practice. These problems highlight some  reasons why theoretical work in this area is hardly applied in practical systems.

In this context, a novel distributed frequency recovery controller is proposed that only needs to be implemented on controllable generator buses. To this end, a network-preserving power system model is adopted. This work is an extension of our former work \cite{Pang:Optimal}. However, in this paper, we design a totally different controller considering a third-order nonlinear generator model with excitation voltage dynamics and nonlinear power flow. 
It is also motivated partly by \cite{stegink:stabilization}, which adopts a similar model, although our results are significantly different from those in \cite{stegink:stabilization}. 
Differing from \cite{Pang:Optimal}, the controller avoids load measurement, which greatly facilitates implementation. 
By using LaSalle's invariance principle, it is proved that the closed-loop system converges to an equilibrium point that solves the economic dispatch problem.  The salient features of the proposed distributed optimal frequency controller are: 
\begin{enumerate}
	\item \emph{Model:} The network-preserving model of power system is used, including excitation voltage dynamics and nonlinear power flow. This, unlike work on the linear model, returns a valid controller even under large disturbances. 
	
	\item \emph{Controllability:} We allow an arbitrary subset of generator buses to be controllable.
	
	\item \emph{Controller:} The distributed  controller achieves the optimal solution to economic dispatch while restoring the nominal frequency, provided that certain sufficient conditions on  active power dynamics of uncontrollable generators and excitation voltage dynamics of all generators are satisfied. 
	
	\item \emph{Communication:} Communication is required between neighboring controllable generators only, and the communication network can be  arbitrary as long as it remains connected. 
	
	\item \emph{Measurement:} No load measurement is needed, and the controller is adaptive to unknown load changes.
\end{enumerate}

The rest of this paper is organized as follows. In Section II, we introduce the power system model. Section III formulates the optimal economic dispatch problem. The distributed controller is proposed in Section IV, and we further prove the optimality and stability of the corresponding equilibrium point in Section V.  The load estimation method is proposed in Sectin VI. We confirm the performance of  controllers via simulations on a detailed power system in Section VII. Section VIII concludes the paper.

\section{Power System Models}

\subsection{Power Network}

Power systems are composed of many generators and loads, which are integrated in different buses and connected by power lines, forming a power network. Buses are divided into three types, controllable generator buses, uncontrollable generator buses and pure load buses. Denote controllable generator buses as $\mathcal {N}_{CG}=\{1, 2, \dots, n_{CG}\}$, uncontrollable generator buses as $\mathcal {N}_{UG}=\{n_{CG}+1, n_{CG}+2, \dots, n_{CG}+n_{UG}\}$, and pure load buses as $\mathcal {N}_{L}=\{n_{CG}+n_{UG}+1, \dots, n_{CG}+n_{UG}+n_{L}\}$. 
Then the set of generator buses is $\mathcal{N}_{G}=\mathcal {N}_{CG} \cup \mathcal {N}_{UG}$ and set of all the buses is $\mathcal {N}=\mathcal{N}_{G}\cup\mathcal{N}_{L}$. It should be noted that load can be connected to any bus besides pure load buses. 

Let $\mathcal E\subseteq \mathcal N\times \mathcal N$ be the set of lines, where $(i,j)\in \mathcal E$ if buses $i$ and $j$ are connected directly. Then the whole system is modeled as a connected graph $\mathcal {G}=(\mathcal {N},\mathcal {E})$. 
The admittance of each line is {$Y_{ij}:=G_{ij}+\sqrt{-1}B_{ij}$} with $G_{ij}=0$ for every line. Denote the bus voltage by $V_i\angle\theta_i$, where $V_i$ is the amplitude and $\theta_i$ is the voltage phase angle. The active and reactive power $P_{ij},Q_{ij}$ from bus $i$ to bus $j$ is
\begin{subequations}
	\label{line power}
	\begin{align}
	P_{ij}&={{V_i}{V_j}{B_{ij}}\sin \left( {{\theta _i} - {\theta _j}} \right)}
	\label{line_reactivepower}  \\ 
	Q_{ij}&={B_{ij}V_i^2 - {V_i}{V_j}B_{ij}\cos \left( {{\theta _i} - {\theta _j}} \right)} 
	\label{line_activepower}
	\end{align}
\end{subequations}

\begin{figure}[t]
	\centering
	\includegraphics[width=0.35\textwidth]{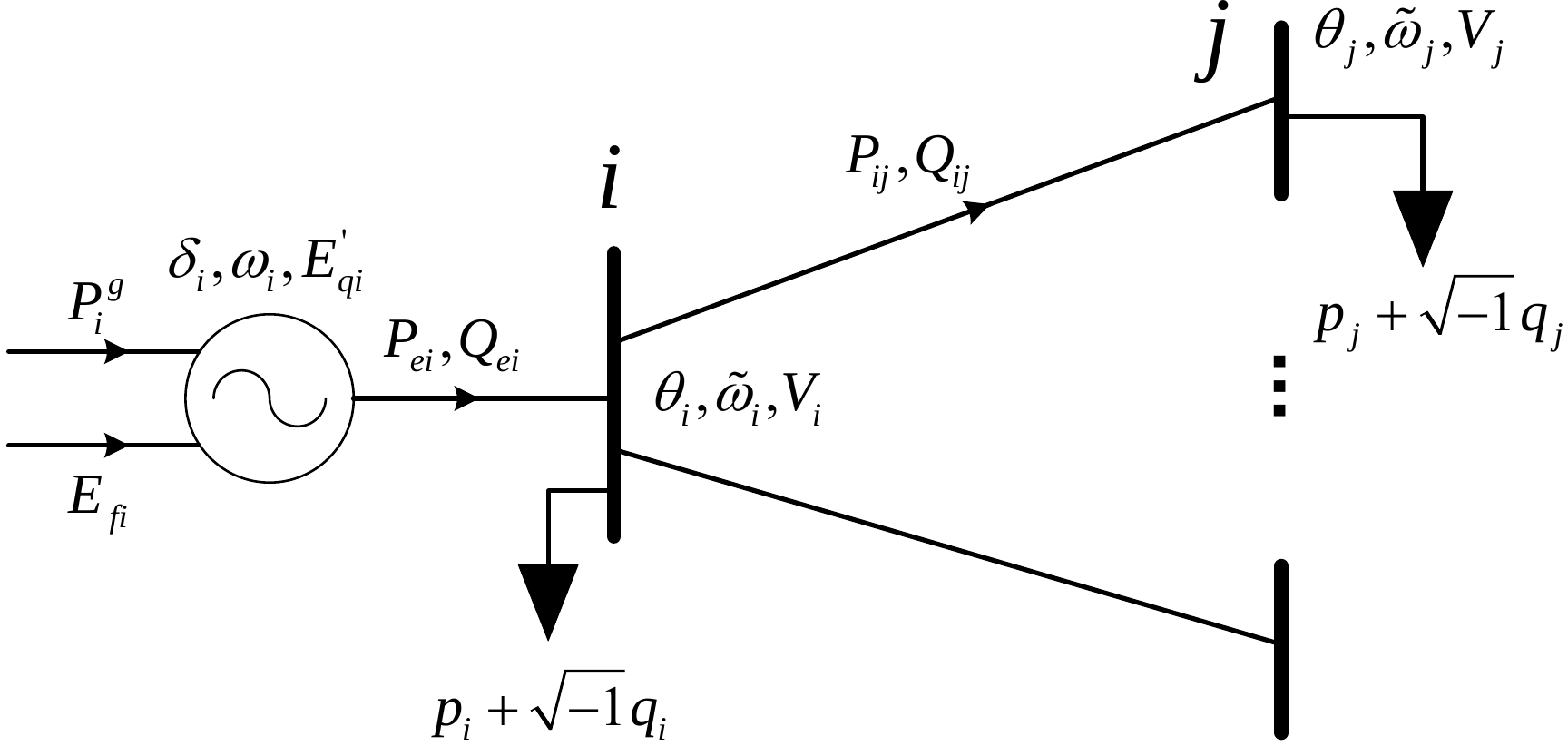}
	\caption{Summary of notations}
	\label{fig:notationsummary}
\end{figure}

For convenience, most notations are summarized in Fig.\ref{fig:notationsummary}. 

\subsection{Synchronous Generators }
For $i\in \mathcal {N}_{G}$, we use the standard third-order generator model (e.g. 
\cite{Stegink:aunifying, Kundur:Power, LIU:Decentralized}) \eqref{eq:SGmodel.1a}-\eqref{eq:SGmodel.1c}. Here \eqref{eq:SGmodel.1d} is the simplified governor-turbine model, and \eqref{eq:SGmodel.1e} is the excitation voltage control model:
\begin{subequations}
	\label{eq:SGmodel}
	\begin{align}
	\dot \delta_i & =  \omega_i
	\label{eq:SGmodel.1a}
	\\
	\dot \omega_i & =   (P^g_i -D_i \omega_i- P_{ei} )/{M_i}
	\label{eq:SGmodel.1b}
	\\
	\dot E^{'}_{qi} & = - {E_{qi}}/{T^{'}_{d0i}} + {E_{fi}}/{T^{'}_{d0i}}
	\label{eq:SGmodel.1c} 
	\\
	\label{eq:SGmodel.1d}
	\dot P_i^g &= - {P_i^g}/{T_i}  + u_i^g 
	\\
	\label{eq:SGmodel.1e}
	\dot E_{fi} & = h(E_{fi},E_{qi})
	\end{align}
	\label{eq:SGmodel.1}
\end{subequations}
In this model, $M_i$ is the moment of inertia; $D_i$ the damping constant;  $T^{'}_{d0i}$  the $d$-axis transient time constant; ${T_i}$  the 
time constant of turbine; $\delta_i$  the power angle of generator $i$; $\omega_i$  the generator frequency deviation compared to an steady state value; $P^g_i$  the mechanical power input; $p_j$  the active load demand; $P_{ei}$  the active power injected to network; $E^{'}_{qi}$ the $q$-axis transient internal voltage; $E_{qi}$  the $q$-axis internal voltage; $E_{fi}$ the excitation voltage. 
$E_{qi}$ is given by
\begin{align}
	E_{qi}=\frac{x_{di}}{x^{'}_{di}}E^{'}_{qi}-\frac{x_{di}-x_{di}^{'}}{x^{'}_{di}}V_i\cos(\delta_i-\theta_i)
\end{align}
where $x_{di}$ is the $d$-axis synchronous reactance, and $x_{di}^{'}$ is the $d$-axis transient reactance.

The active and reactive power (denoted by $Q_{ei}$) injection to the network are
\begin{subequations}
	\label{eq:Model.PQ}	
	\begin{align}
		\label{Pei}
		&P_{ei}=\frac{E^{'}_{qi}V_i}{x_{di}^{'}}\sin(\delta_i-\theta_i)\\
		\label{Qei}
		&Q_{ei}=\frac{V_i^2}{x^{'}_{di}}-\frac{E^{'}_{qi}V_i}{x_{di}^{'}}\cos(\delta_i-\theta_i)
	\end{align}
\end{subequations}

For controllable generators $i\in \mathcal {N}_{CG}$, the capacity limits are 
\begin{align}
	\underline{P}_i^g \le {P}_i^g \le \overline{P}_i^g
\end{align}
where $\underline{P}_i^g, \overline{P}_i^g$ are lower and upper limits of ${P}_i^g$.

\subsection{Dynamics of Voltage Phase Angles}
To build a network-preserving power system model,  relation between generators and power network  should be explicitly established. In this paper, loads for bus $i\in\mathcal {N}$ are simply modeled as constant active and reactive power injections. Then  the following equations are used to dictate the power balance and voltage phase-angle dynamics at each bus:
\begin{subequations}
	\label{load power}
	\begin{align}
		\dot{\theta}_{i}&=\tilde\omega_i,\ \  i\in\mathcal {N} \\
		\label{load activepower}
		0&=P_{ei}-\tilde D_i\tilde\omega_i-p_{i}-\sum\nolimits_{j \in N_i} {P_{ij} }, \ \ i\in \mathcal{N}_{G}   
		\\
		\label{load activepower2}
		0&=-\tilde D_i\tilde\omega_i-p_{i}-\sum\nolimits_{j \in N_i} {P_{ij} }, \ \ i\in \mathcal{N}_{L}   \\  
		0&=Q_{ei}-q_{i}-{ \sum\nolimits_{j \in N_i}Q_{ij} }, \ \ i\in \mathcal{N}_{G} \\
		0&=-q_{i}-{ \sum\nolimits_{j \in N_i}Q_{ij} } , \ \ i\in \mathcal{N}_{L}
		\label{reactive_power}
	\end{align}
\end{subequations}
where, $p_{i}, q_{i}$ are active and reactive load demands, respectively; $\tilde{\omega}_i$  the frequency deviation at bus $i$; $N_i$  the set of buses connected directly to bus $i$; $\tilde D_i$ the damping constant at  bus $i$; $\tilde D_i\tilde\omega_i$  the change of  frequency-sensitive load \cite{Changhong:Design}. 

In power system, line power flows are mainly related to power angle difference between two buses rather than the power angles independently. Then, we define new variables to denote angle differences as $\eta_{ii}:=\delta_i-\theta_i, i \in {\cal N}_{G}$ and $\eta_{ij}:=\theta_i-\theta_j, i, j \in {\cal N}$. The time derivative of $\eta_{ii}$ and $\eta_{ij}$ are 
\begin{subequations}
	\begin{align}
		\dot \eta_{ii} & =  \omega_i - \tilde{\omega}_i, \ i \in {\cal N}_{G}
		\label{eq:closedloop.1a1}
		\\
		\dot \eta_{ij} & =  \tilde{\omega}_i - \tilde{\omega}_j, \ i, j \in {\cal N}, i\neq j
		\label{eq:closedloop.1a2}
	\end{align}	
	respectively. In the following analysis, we use $\eta_{ii}$ and $\eta_{ij}$ as new state variables instead of $\delta_i$ and $\theta_{i}$.
		\label{eq:Model.eta}
\end{subequations}

To  summarize, $\eqref{line power} - \eqref{eq:Model.PQ}, \eqref{load activepower}-\eqref{reactive_power}, \eqref{eq:closedloop.1a1}, \eqref{eq:closedloop.1a2}$ constitute the nonlinear network preserving model of power systems, which is in a form of differential-algebraic equations (DAE).

\subsection{Communication Network}
In this paper, we consider a communication graph among the buses of controllable generators only. Denote $E\subseteq \mathcal {N}_{CG} \times \mathcal {N}_{CG}$ as the set of communication links. 
If generator $i$ and $j$ can communicate directly to each other, we denote $(i,j)\in E$. Obviously, edges in the communication graph $E$ can be different from lines in the power network $\mathcal{E}$. For the communication network, we make the following assumption:
\begin{itemize}
	\item [\textbf{A1:}] The communication graph $E$ is undirected and connected.
\end{itemize}

\section{Formulation of Optimal Frequency Control}
\subsection{Optimal Power-Sharing Problem in Frequency Control}
The purpose of optimal frequency control is to let all the controllable generators share power mismatch economically when restoring frequency. Then we have the following optimization formulation, denoted by SFC.
\begin{subequations}
	\label{SFC}
	\begin{align}
		\label{objective function}
		\text{SFC:}\ \ \ &\mathop {\min }\limits_{P_i^g, i\in {\cal {N}}_{CG}} 
		\ \ \sum\nolimits_{i\in {\cal {N}}_{CG}} f_i(P_i^g) \\
		\label{eq:power balance}
		&\text{s.t.}\quad \sum\limits_{i\in {\cal {N}}_{CG}} P_i^g = \sum\limits_{i\in \cal {N}} p_i - \sum\limits_{i\in {\cal {N}}_{UG}} P_i^{g*}\\
		& \quad\quad\ \  \underline{P}_i^g \le {P}_i^g \le \overline{P}_i^g, \quad i\in {\cal {N}}_{CG}
	\end{align}
\end{subequations}
where $P_i^{g*}$ is the mechanical power of uncontrollable generator in the steady state.
In (\ref{objective function}), $f_i(P_i^g)$ concerns the controllable generation $P_i^g$, satisfying the following assumption:
\begin{itemize}
	\item [\textbf{A2:}] The objective $f_i(P_i^g)$ is second-order continuously differentiable, strongly convex and $f^{'}_i(P_i^g)$ is Lipschitz continuous with Lipschitz constant $l_i>0$. i.e. $\exists\ \alpha_i >0, \alpha_i\le f_i^{''}(P_i^g)\le l_i $.
\end{itemize}

\textcolor{black}{To ensure} the feasibility of the optimization problem, we make an additional assumption.
{\color{black}
\begin{itemize}
	\item [\textbf{A3:}]  The system satisfies
	\begin{align}
		\sum\limits_{i\in {\cal {N}}_{CG}} \underline P_i^g \le \sum\limits_{i\in \cal {N}} p_i  - \sum\limits_{i\in {\cal {N}}_{UG}} P_i^{g*} \le \sum\limits_{i\in {\cal {N}}_{CG}} \overline P_i^g
		\label{eq:A3}
	\end{align} 
\end{itemize}
Specifically, we say A3 is \emph{strictly satisfied} if all the inequalities  in \eqref{eq:A3} \emph{strictly} hold.}
\subsection{Equivalent Optimization Model with Virtual Load Demands}
In (\ref{eq:power balance}), load demands are injected to every bus, which sometimes cannot be measured accurately if at all. As a consequence, the values of $p_i$ may be unknown to both the controllable generators $i$, $i\in {\cal {N}}_{CG}$ and the uncontrollable generators $i$, $i\in {\cal {N}}_{UG}$. To circumvent such an obstacle in design, we introduce a set of new variables, $\hat p_i$, to re-formulate SFC as the following equivalent problem:
\begin{subequations}
	\label{ESFC}
	\begin{align}
	\label{objective function2}
	\text{ESFC:}\ \ \ &\mathop {\min }\limits_{P_i^g, i\in {\cal {N}}_{CG}}\ \ \sum\nolimits_{i\in {\cal {N}}_{CG}} f_i(P_i^g) \\
	\label{eq:power balance2}
	&\text{s.t.}\quad \sum\nolimits_{i\in {\cal {N}}_{CG}} P_i^g = \sum\nolimits_{i\in {\cal {N}}_{CG}}\hat p_i \\
	& \quad\quad\ \  \underline{P}_i^g \le {P}_i^g \le \overline{P}_i^g, \quad i\in {\cal {N}}_{CG}
	\end{align}
\end{subequations}
where $\hat p_i$ is the \textit{virtual load demand} supplied by generator $i$ in the steady state, which is a constant, satisfying 
$\sum\nolimits_{i\in {\cal {N}}_{CG}}\hat p_i = \sum\nolimits_{i\in \cal {N}} p_i- \sum\nolimits_{i\in {\cal {N}}_{UG}} P_i^{g*}$.  Obviously, the number of virtual loads should be equal to that of the controllable generators. 

Note that the power balance constraint \eqref{eq:power balance} only requires
that all the generators supply all the loads while it is not necessary to figure out which loads are supplied exactly by which generators. Hence we treat virtual load demands $\hat p_i$ as the effective demands 
supplied by generator $i$ for dealing  with the issue that only a part of generators are controllable. 

Simply letting $\sum\nolimits_{i\in {\cal {N}}_{CG}}\hat p_i = \sum\nolimits_{i\in \cal {N}} p_i- \sum\nolimits_{i\in {\cal {N}}_{UG}} P_i^{g*}$, we immediately have the following Lemma:
\begin{lemma}
	\label{equivalent}
	The problems SFC \eqref{SFC} and ESFC \eqref{ESFC}
	have the same optimal solutions.
\end{lemma}

\section{Controller Design}
\subsection{Distributed Frequency Control of Controllable Generators}

\subsubsection{Controller design based on  primal-dual gradient algorithm}
Invoking the primal-dual gradient algorithm, the Lagrangian of the ESFC \eqref{ESFC} is given by
\begin{align}
	L=&\sum\nolimits_{i\in {\cal {N}}_{CG}} f_i(P_i^{g})  + \mu\left(\sum\nolimits_{i\in {\cal {N}}_{CG}} P_i^g - \sum\nolimits_{i\in {\cal {N}}_{CG}}\hat p_i\right) \nonumber\\
	& + \gamma _i^-(\underline{P}_i^g - {P}_i^g ) + \gamma _i^+ ( {P}_i^g - \overline{P}_i^g), \quad\quad i\in {\cal {N}}_{CG}
\end{align}
where $\mu, \gamma _i^-, \gamma _i^+$ are Lagrangian multipliers. Based on   primal-dual update, the controller for $i\in {\cal {N}}_{CG}$ is designed as
\begin{subequations}
	\label{eq:controller0}	
	\begin{align}
	\label{eq:controlmodel0}
	u^g_i & ={P^g_i}/{T_i}  -{k_{P_i^g}}\left(\omega_i  +  (f_i^{'}(P_i^g)+\mu-\gamma _i^- + \gamma _i^+)\right) \\
	\dot \mu & = k_{\mu} \bigg(\sum\nolimits_{i\in {\cal {N}}_{CG}} P_i^g - \sum\nolimits_{i\in {\cal {N}}_{CG}}\hat p_i \bigg) 	
	\label{eq:hatpi0} \\
	\label{controller_d0}
	\dot \gamma _i^- & = k_{\gamma_i} \left[\underline{P}_i^g - {P}_i^g \right]_{\gamma _i^-}^+ \\
	\label{controller_e0}
	\dot \gamma _i^+ & = k_{\gamma_i}\left[{P}_i^g - \overline{P}_i^g \right]_{\gamma _i^+}^+
	\end{align}
\end{subequations}
where $k_{P_i^g}, k_{\mu}, k_{\gamma_i}$ are positive constants; $u^g_i$ the control input; $f_i^{'}(P_i^g)$ the marginal cost at $P_i^g$. For any $x_i$, $a_i \in \mathbb R$, the operator is defined as: $[x_i]^+_{a_i}=x_i$ if $a_i>0\  \text{or}\ x_i>0$; and $[x_i]^+_{a_i}=0$ otherwise.
\subsubsection{Estimating $\mu$ by second-order consensus} In \eqref{eq:hatpi0}, $\mu$ is a global variable, which is a function of mechanical powers and loads of the entire system. To circumvent the obstacle in implementation, a second-order consensus based method is utilized to estimate $\mu$ locally by using neighboring information only. Specifically, for $i\in \mathcal N_{CG}$, the controller  is revised to:
\begin{subequations}
	\label{eq:controller}	
	\begin{align}
		\label{eq:controlmodel}
		u^g_i & ={P^g_i}/{T_i}  -{k_{P_i^g}}\left(\omega_i  +  f_i^{'}(P_i^g)+\mu_i-\gamma _i^- + \gamma _i^+\right) \\
		\dot \mu_i & = k_{\mu_i} \bigg(P_i^g - \hat p_i - \sum\limits_{j \in {N_{ci}}} {\left( {{{ \mu }_i} - {{\mu }_j}} \right)}  - \sum\limits_{j \in {N_{ci}}} {{z_{ij}}}   \bigg) 	
		\label{controller_b} \\
		\label{controller_c}
		\dot z_{ij} &= k_{z_i} (\mu_i-\mu_j) \\
		\label{controller_d}
		\dot \gamma _i^- & = k_{\gamma_i} \left[\underline{P}_i^g - {P}_i^g \right]_{\gamma _i^-}^+ \\
		\label{controller_e}
		\dot \gamma _i^+ & = k_{\gamma_i}\left[{P}_i^g - \overline{P}_i^g \right]_{\gamma _i^+}^+
	\end{align}
\end{subequations} 
where, $ k_{\mu_i}, k_{z_i}$ are positive constants; $N_{ci}$  the set of  neighbors of bus $i$ in the communication graph; 
$\mu_i$ the local estimation of $\mu$. Here, \eqref{controller_b} and \eqref{controller_c} are used to estimate $\mu$ locally, where only neighboring information is needed. $z_{ij}$ is an auxiliary variable to guarantee the consistency of all $\mu_i$.

For the Larangian multiplier $\mu$, $-\mu$ is often regarded as the marginal cost of generation. Theoretically, $-\mu_i$ should reach consensus for all the generators in the steady state. Since $\dot \mu_i = 0$ holds in the steady state, we have $P_i^g - \hat p_i - \sum\nolimits_{j \in {N_{ci}}} {{z_{ij}}}  =0$. Hence, $z_{ij}$ can be regarded as the virtual line power flow of edge $(i,j)$ in the communication graph.

\subsection{Active Power Dynamics of Uncontrollable Generators}
To guarantee  system stability, a sufficient condition is given for the active power dynamics of uncontrollable generators. 
\begin{itemize}
	\item [\textbf{C1:}] The active power dynamics of uncontrollable generators are strictly incrementally output passive in terms of the input $-\omega_i$ and output $P_{i}^g$, i.e., there exists a continuously differentiable, positive semidefinite function $S_{\omega_i}$ such that  
	\begin{align}
	\dot S_{\omega_i} \le \left(-\omega_i-(-\omega_i^*)\right)\left(P_{i}^g-P_{i}^{g*}\right)  - \phi_{\omega_i}\ (P_{i}^g-P_{i}^{g*}) \nonumber
	\end{align}
	where $\phi_{\omega_i}$ is a positive definite function, and $\phi_{\omega_i}=0$ holds only when $P_{i}^g=P_{i}^{g*}$.
\end{itemize}
The condition C1 on the active power dynamics of uncontrollable generators is easy to verify. As an example,  the  commonly-used primary frequency controller 
\begin{align}
\label{droopcontrol}
	u_i^g=-\omega_i+\omega_i^*-k_{\omega_i}(P_i^g-P_i^{g*})+{P^g_i}/{T_i}
\end{align}
 satisfies C1 whenever $k_{\omega_i}>0$. In this case, we have $S_{\omega_i}=\frac{k_1}{2}(P_i^g-P_i^{g*})^2$ with $k_1>0$ and $\phi=k_2(P_i^g-P_i^{g*})^2$ with $0<k_2\le k_{\omega_i} \cdot k_1$.

\subsection{Excitation Voltage Dynamics of All Generators}
Similar to the uncontrollable generators, the following sufficient condition on excitation voltage dynamics of all generators is needed to guarantee system stability, since we do not design specific excitation voltage controllers here.
\begin{itemize}
	\item [\textbf{C2:}] The excitation voltage dynamics are strictly incrementally output passive in terms of the input $-E_{qi}$ and output $E_{fi}$, i.e., there exists continuously differentiable, positive semidefinite function $S_{E_i}$ such that  
	\begin{align}
	\dot S_{E_i} \le \left(-E_{qi}-(-E^*_{qi})\right)\left(E_{fi}-E^*_{fi}\right)  - \phi_{E_i}\ (E_{fi}-E^*_{fi}) \nonumber
	\end{align}
	where $\phi_{E_i}$ is a positive definite function, and $\phi_{E_i}=0$ holds only when $E_{fi}=E^*_{fi}$.
\end{itemize}
C2 is also easy to satisfy. As an example, it can be verified that the controller given in \cite{LIU:Decentralized}
\begin{align}
\label{voltage controller}
	h(E_{fi},E_{qi})=-E_{fi}+E^*_{fi}-k_{E_i}(E_{qi}-E^*_{qi})
\end{align}
with $k_{E_i}>0$ satisfies C2. In this case, $S_{E_i}=\frac{k_3}{2}(E_{fi}-E^*_{fi})^2$ with $k_3>0$ and $\phi=k_4(E_{fi}-E^*_{fi})^2$ with $0<k_4\le k_{E_i} \cdot k_3$.

\section{Optimality and Stability }
After implementing the controller on the physical power system, the closed-loop system reads
\begin{align}
\label{eq:closedloop}
\left\{ \begin{array}{l}
\eqref{line power} - \eqref{eq:Model.PQ}, 
\eqref{load activepower}-\eqref{reactive_power}, 
\eqref{eq:closedloop.1a1},\ \eqref{eq:closedloop.1a2} \\
\eqref{eq:controlmodel}-\eqref{controller_e}
\end{array} \right.
\end{align}

In this section, we  prove the optimality and stability of the closed-loop system \eqref{eq:closedloop}.
\subsection{Optimality}
Denote the trajectory of closed-loop system as $v(t)=\left(\eta(t), \omega(t), \tilde{\omega}(t), P^{g}(t), \mu(t), z(t), \gamma^-(t), \gamma^+(t), E_{q}^{'}(t), V(t)\right)$. 
Define the equilibrium set of \eqref{eq:closedloop} as
\begin{align}
{\cal V}:=\{ v^*|v^* \text{is an equilibrium of}\ \eqref{eq:closedloop} \}
\end{align}
We first present the following Theorem.
\begin{theorem}
	\label{thm:optimality}
	Suppose assumptions A1, A2 and A3 hold. In equilibrium of \eqref{eq:closedloop}, following assertions are true.
	\begin{enumerate}
		\item The mechanical powers $P_i^g$ satisfy $\underline{P}_i^g \le {P}_i^{g*} \le \overline{P}_i^g $,  $\forall i \in \mathcal {N}_{CG}$. 
		\item System frequency recovers to the nominal value, i.e. $\omega_i^*=0, \forall i \in \mathcal {N}_{CG} \cup \mathcal {N}_{UG}$, and $\tilde\omega_i^*=0, \forall i \in \cal N$.
		\item The marginal generation costs satisfy $f_i^{'}(P_i^{g*}) - \gamma _i^{-*} + \gamma _i^{+*} = f_j^{'}(P_j^{g*}) - \gamma _j^{-*} + \gamma _j^{+*}, i,j \in {\cal N}_{CG}$.
		\item $P_i^{g*}$ is the unique optimal solution of SFC problem (\ref{SFC}).
		\item  $\mu_i^*$ is unique if A3 is strictly satisfied. 
	\end{enumerate}
\end{theorem}


The detailed proof of Theorem \ref{thm:optimality} is given in Appendix.A of this paper. It shows that the nominal frequency is recovered at equilibrium, and marginal generation costs are identical for all controllable generators, implying the optimality of any equilibrium. 

\subsection{Stability}
In this section, the stability of the closed-loop system (\ref{eq:closedloop}) is proved. 
First we define a function as
\begin{align}
\label{Lagrangian L}
\hat{L}&: = \sum\nolimits_{i\in {\cal {N}}_{CG}} f_i(P_i^{g*})  + \sum\nolimits_{i\in {\cal {N}}_{CG}} \mu_i (P_i^g - \hat p_i ) \nonumber\\
&- \sum\nolimits_{i \in {{\cal {N}}_{CG}}} {{{ \mu }_i}{z_{ij}}}  - \frac{1}{2}\sum\nolimits_{i \in {\cal {N}}_{CG}} \left({ \mu }_i\sum\nolimits_{j \in {N_i}} {\left( {{{ \mu }_i} - {{ \mu }_j}} \right)} \right) \nonumber\\
& + \sum\limits_{i\in {\cal {N}}_{CG}} \gamma _i^-(\underline{P}_i^g - {P}_i^g ) + \sum\limits_{i\in {\cal {N}}_{CG}} \gamma _i^+ ( {P}_i^g - \overline{P}_i^g)
\end{align}
Denote $x_1 := ( P^g)$, $x_2 := ( \mu, z, \gamma_i^{-}, \gamma_i^{+})$, $x := (x_1, x_2)$. Then $\hat L(x_1, x_2)$ is convex in $x_1$ and concave in $x_2$.

Before giving the main result, we construct a Lyapunov candidate function composed of four parts: the quadratic part, the potential energy part, conditions C1 and C2 related parts, as we now explain.

For $i\in {\cal N}_G$, the quadratic part is given by 
\begin{align}
\label{kinetic function}
W_k(\omega, x)=&\sum\limits_{i \in {\cal N}_{G}}\frac{1}{2}M_i(\omega_i-\omega_i^*)^2  + \frac{1 }{2} (x-x^*)^TK^{-1}(x-x^*)
\end{align}
where $K=\text{diag}(k_{P_i^g}, k_{\mu_i}, k_{z_i}, k_{\gamma_i})$ is a diagonal positive definite matrix.

Denoting $x_p=(E_{qi}^{'},V_i,\delta_i,\theta_i)$, the potential energy part is 
\begin{equation}
\label{potential function}
W_p(x_p)=\tilde W_p(x_p)-(x_p-x_p^*)^T\nabla_{x_p}\tilde W_p(x_p^*)-\tilde W_p(x_p^*)
\end{equation}
where, 
\begin{align}
\label{potential function0}
&\tilde W_p(E_{qi}^{'},E_{i},V_i,\delta_i,\theta_i)=\sum\nolimits_{i \in {\cal N}} \frac{1}{2}B_{ii}V_i^2 + \sum\nolimits_{i \in \cal {\cal N}} p_i\theta_i \nonumber\\
&\ - \sum\limits_{i \in {\cal {N}}} q_i\text{ln}\ V_i - \frac{1}{2}\sum\limits_{i \in {\cal N}}\sum\limits_{j \in N_i} {{V_i}{V_j}{B_{ij}} \cos \left( {{\theta _i} - {\theta _j}} \right)} \\
&\ -\sum\limits_{i \in {\cal N}_{G}} \frac{E^{'}_{qi}V_i}{x_{di}^{'}}\cos(\delta_i-\theta_i)  +\sum\limits_{i \in {\cal N}_{G}}\frac{x_{di}}{2x^{'}_{di}(x_{di}-x_{di}^{'})}\left(E^{'}_{qi}\right)^2  \nonumber
\end{align}

Conditions C1 and C2 related parts are $\sum\nolimits_{i\in {\cal {N}}_{UG}} S_{\omega_i}$ and $\sum\nolimits_{i\in {\cal {N}}_{G}}\frac{1}{{T_{d0i}^{'}(x_{di}-x_{di}^{'})}} S_{E_i}$ respectively.

The Lyapunov function is defined as 
\begin{align}
\label{Lyapunov}
W=&W_k+W_p  +\sum\limits_{i\in {\cal {N}}_{UG}} S_{\omega_i} + \sum\limits_{i\in {\cal {N}}_{G}}\frac{S_{E_i}}{{T_{d0i}^{'}(x_{di}-x_{di}^{'})}} 
\end{align}

Then, we give the following assumption. 

{\color{black}
	\noindent\textbf{A4:} The Hessian of $W_p$ satisfies $\nabla^2_v W_p(v)>0$ at desired equilbirium.
}

Since the voltage phase deviation between two neighboring buses  is not large in practice, A4 is usually satisfied. Detailed explanations can be found in Appendix.B of this paper.

The following stability result can be obtained.
\begin{theorem}
	\label{thm:stability}
	Suppose A1--A4 and C1, C2 hold. For every $v^*$, there exists a neighborhood $\cal S$ around $v^*$ where all trajectories $v(t)$	satisfying \eqref{eq:closedloop} starting in $\cal S$ converge to the set ${\cal V}$. In addition,  each  trajectory converges to an equilibrium point.
\end{theorem}

We can further prove that $\nabla^2 W>0$ and $\dot W \le 0$. Moreover,  $\dot W = 0$ holds only at equilibrium point. Then the theorem can be proved using the LaSalle's invariance principle \cite[Theorem 4.4]{Khalil:Nonlinear}. 
Details of the proof are given in  Appendix.B.

\section{Implementation Based on Frequency Measurement}
\subsection{Estimation and Optimality}
Note that virtual load demands $\hat p_i$ used in the controller \eqref{eq:controller} are difficult to directly measure or estimate in practice. Lemma \ref{equivalent} implies that any $\hat p_i$ are valid  as long as $\sum\nolimits_{i\in {\cal {N}}_{CG}}\hat p_i = \sum\nolimits_{i\in \cal {N}} p_i - \sum\nolimits_{i\in {\cal {N}}_{UG}} P_i^{g*}$. Noticing that the power imbalance is very small in  normal operation,  we have $\sum\nolimits_{i\in {\cal {N}}_{CG}}P_{ei} \approx \sum\nolimits_{i\in \cal {N}} \hat p_i$. In fact, they are identical in steady state. Hence, we specify $\hat p_i=P_{ei}$, which implies $P_i^g-\hat p_i= P_i^g-P_{ei}=M_i\dot\omega_i+D_i\omega_i$. That leads to an estimation algorithm of $\mu_i$ 
\begin{align}
\label{load estimation}
\dot \mu_i & = k_{\mu_i} \big(- \sum\limits_{j \in {N_i}} {\left( {{{ \mu }_i} - {{\mu }_j}} \right)}  - \sum\limits_{j \in {N_i}} {{z_{ij}}} + M_i\dot\omega_i+D_i\omega_i \nonumber\\
           &\qquad\qquad +\tau_i (-\mu_i - f_i^{'}(P_i^g)+\gamma _i^- - \gamma _i^+)
\big)
\end{align}
where $0<\tau_i<4/l_i$.
This way, we only need to measure frequencies $\omega_i$ at each bus locally, other than the global load demands.  
Since the controller only needs $\mu_i$ of neighboring buses in the communication graph, it is easy to implement.  

{\color{black}Now, we reconstruct the closed-loop system by replacing  \eqref{controller_b} with \eqref{load estimation} in \eqref{eq:closedloop}, which is 
\begin{align}
	\label{eq:closedloop2}
	\left\{ \begin{array}{l}
	\eqref{line power} - \eqref{eq:Model.PQ}, 
	\eqref{load activepower}-\eqref{reactive_power}, 
	\eqref{eq:closedloop.1a1},\ \eqref{eq:closedloop.1a2} \\
	\eqref{eq:controlmodel}, \eqref{controller_c} -\eqref{controller_e}, \eqref{load estimation}
	\end{array} \right.
\end{align}
We have the following lemma.
\begin{lemma}
	\label{Lemma:optimality}
	Assertions 1)-5) in Theorem \ref{thm:optimality} still hold for the equilibrium of \eqref{eq:closedloop2}.
\end{lemma}

The proof of Lemma \ref{Lemma:optimality} is given in Appendix.C.

\subsection{Discussion on Stability}
Recall \eqref{eq:SGmodel.1b}, then  \eqref{load estimation} is derived to
\begin{align}
	\label{load estimation2}
	\dot \mu_i & = k_{\mu_i} \big( P_i^g -\hat p_i + \hat p_i - P_{ei} - \sum\limits_{j \in {N_i}} {\left( {{{ \mu }_i} - {{\mu }_j}} \right)}  - \sum\limits_{j \in {N_i}} {{z_{ij}}}  \nonumber\\
	&\qquad\qquad +\tau_i (-\mu_i - f_i^{'}(P_i^g)+\gamma _i^- - \gamma _i^+)
	\big)
\end{align}
Denote $\rho_i = \hat p_i - P_{ei}= P_{ei}^* - P_{ei}$, which is the difference of electric power and its value in the steady state. We have the following assumption
\begin{itemize}
	\item [\textbf{A5:}] The disturbance can be written as $\rho_i=\beta_i(t)\omega_i$, where $|\beta_i(t)|\le \bar \beta_i$ and $\bar \beta_i$ is a positive constant. In addition, the set  $\{\ t<\infty\ |\ \omega_i(t)=\omega_i^*\ \}$ has a measure zero.
\end{itemize} Whenever $\omega_i\neq\omega_i^*$, there alway exists such $\beta_i(t)$. A5 argues that $\omega_i(t)=\omega_i^*$ only happens at isolated points except equilibrium. Generally, this is reasonable in power system. 

Denote the state variables of \eqref{eq:closedloop2} and its equilibrium set are $\tilde{v}$ and $\tilde{\cal V}$ respectively. We have following stability result
\begin{theorem}
	\label{thm:stability2}
	Suppose A1--A5, C1, C2 hold and \eqref{eq:A3} is not binding. For every $\tilde v^*$, there exists a neighborhood $\cal S$ around $\tilde v^*$ where all trajectories $\tilde v(t)$ satisfying \eqref{eq:closedloop2} starting in $\cal S$ converge to the set $\tilde {\cal V}$ whenever
	\begin{align}
	\label{beta_range}
		\bar \beta_i<\sqrt{\tau_i D_i(4-\tau_i l_i)}.
	\end{align} 
	Moreover, the convergence of each such trajectory is to a point.
\end{theorem}

The proof of Theorem \ref{thm:stability2} is given in Appendix.D. In fact, the range of $\beta(t)$ can be very large as long as $l_i$ is small enough. For example, set $\tau_i=3/l_i$, then $\bar\beta_i(t)<\sqrt{3D_i/l_i}$. As we know, if we change the objective function to $k\sum\nolimits_{i\in {\cal {N}}_{CG}} f_i(P_i^g), k>0$, the optimal solution will not change. Thus, $l_i$ can be very small as long as $k$ is small enough. In this regard, \eqref{beta_range} is not conservative.
Moreover, we will illustrate in Section \ref{sec:sr} that even \eqref{eq:A3} is binding, our controller still works.

}

\section{Simulation Results}
\label{sec:sr}

\subsection{Test System}
To test the proposed controller, the New England 39-bus system with 10 generators as shown in Fig.\ref{fig:system}, is utilized. In the simulation, we  apply  \eqref{load estimation} to estimate the virtual load demands. 
All simulations are implemented in the commercial power system simulation
software PSCAD \cite{website:PSCAD},
and are carried on a notebook with 8GB memory and 2.39 GHz CPU. 

We control only a subset of these generators, namely G32, G36, G38, G39, while the remaining  are equipped with the primary frequency control  given in \eqref{droopcontrol}.   
  In particular, we apply the controller \eqref{eq:controller} derived based on a simple model to a much more realistic and complicated model in PSCAD. The detailed electromagnetic transient model of three-phase synchronous machines (sixth-order model) is adopted to simulate generators with  governors and exciters. All the lines and transformers take both resistance and reactance into account.  The loads are modeled as  fixed loads in PSCAD. The communication graph is undirected and set as $G32\leftrightarrow{}{}G36\leftrightarrow{}{}G38\leftrightarrow{}{}G39\leftrightarrow{}{}G32$. 

The objective function is set as $f_i=\frac{1}{2}a_i(P_i^g)^2+b_iP_i^g$, which is  the generation cost of generator $i$. Capacity limits of $P_i^g$ and parameters $a_i, b_i$ are given in Table \ref{tab:DistConstraints}. 

The closed-loop system is shown in Fig.\ref{fig:closed_loop_system}, where each  generator only needs to measure local frequency, mechanical power, voltage and phase angle to compute its control command. Note that there is no  load measurement and \textcolor{black}{only  $\mu_i$ are communicated between neighboring controllable generators. }
\begin{figure}[t]
	\centering
	\includegraphics[width=0.48\textwidth]{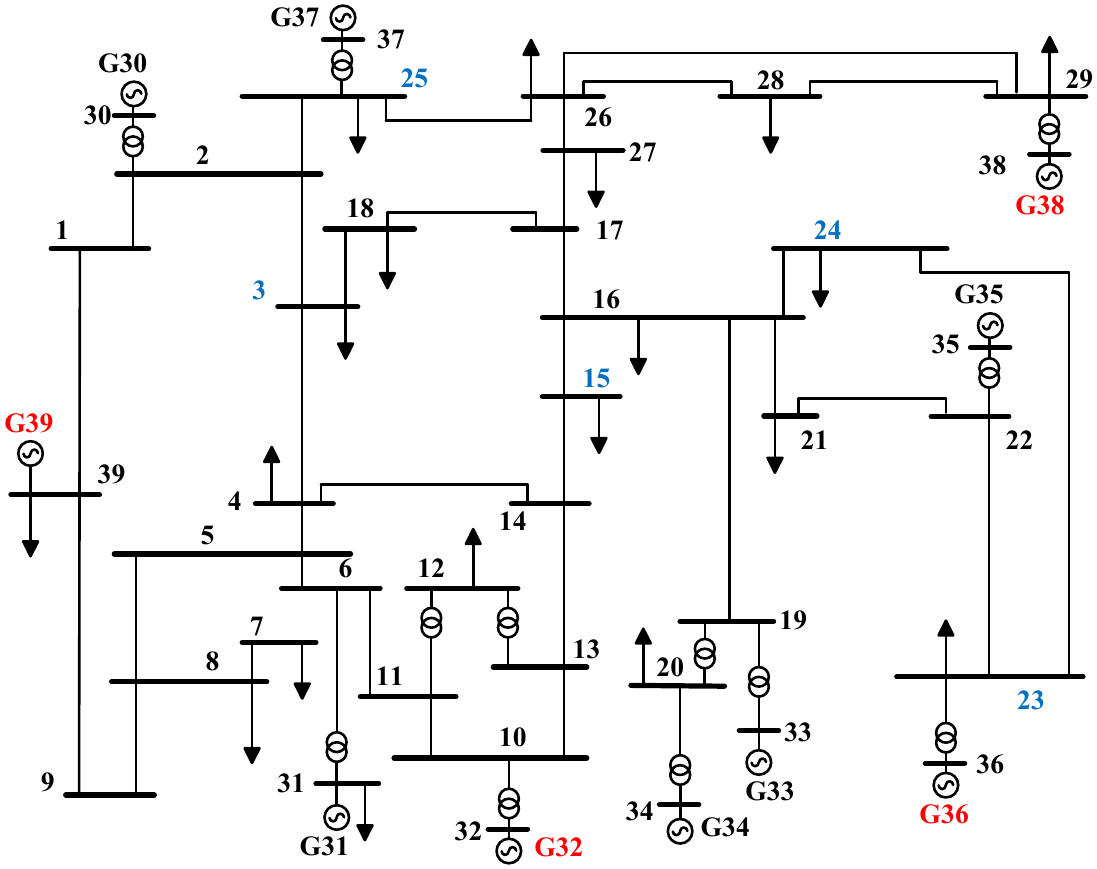}
	\caption{The New England 39-bus system}
	\label{fig:system}
\end{figure}
\begin{figure}[t]
	\centering
	\includegraphics[width=0.45\textwidth]{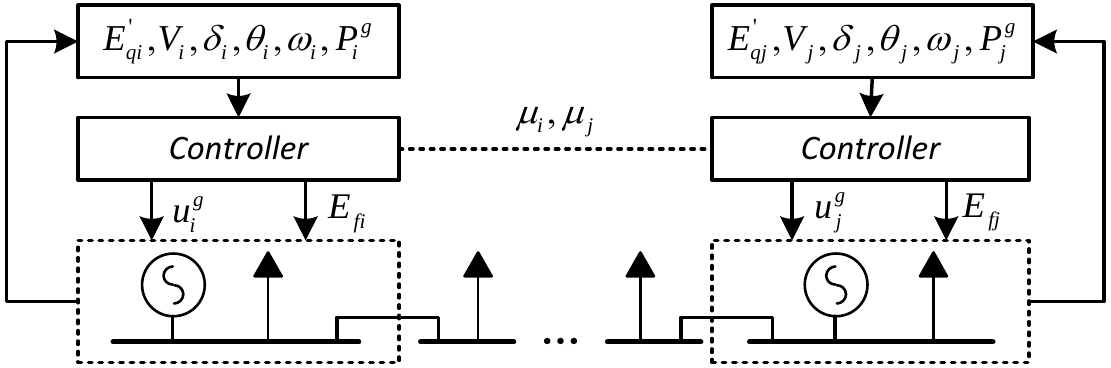}
	\caption{Diagram of the closed-loop system}
	\label{fig:closed_loop_system}
\end{figure}

\begin{table}[!t]
	\centering
	\footnotesize
	\caption{ {Capacity limits of generators }}
	\label{tab:DistConstraints}
	\begin{tabular}{c c c c}
		\hline 
		G$i$ & [$\underline{P}^g_i$, $\overline{P}^g_i$] (MW) & $a_i$ & $b_i$ \\
		\hline
		32   & [0, 1000]   & 0.00009 & 0.032   \\
		36   & [0, 1000]   & 0.00014 & 0.030   \\
		38   & [0, $\ $850]& 0.00010 & 0.032   \\
		39   & [0, 1080]   & 0.00008 & 0.032   \\
		\hline
	\end{tabular}
\end{table}

\subsection{Results under Small Disturbances}
We consider the following scenario: 1) at $t=10$s, there is a step change of $60$MW load demands at each of buses 3, 15, 23, 24, 25; 2) at $t=70$s, there is another step change of $120$MW load at bus $23$. Neither the original load demands nor their changes are known to the generators.

\subsubsection{Equilibrium}
In  steady states, the nominal frequency is well recovered. The optimal mechanical powers are given in Table \ref{tab:eper}, which are identical to the optimal solution of  \eqref{SFC} computed by centralized optimization. Stage 1 is for the period from $10$s to $70$s, and Stage 2 from $70$s to $130$s. {\color{black}The values in Table \ref{tab:eper} are generations at the end of each stage. In Stage 1, no generation reaches its limit, while in Stage 2 both G38 and G39 reach their upper limits. At the end of each stage, the marginal generation cost $-\mu_i$ of generator $i$, converges identically (see Fig. \ref{mu_z}), implying the optimality of the results. } The test results confirm the theoretical analyses and demonstrate that our controller can automatically attain optimal operation points even in the more realistic and sophisicated   model.

\begin{table}[!t]
	\centering
	\footnotesize
	\caption{{Equilibrium points}}
	\label{tab:eper}
	\begin{tabular}{c c c c c}
		\hline 
		& $P_{32}^g$ (MW) &  $P_{36}^g$ (MW)  &  $P_{38}^g$ (MW)  &  $P_{39}^g$ (MW) \\
		\hline
		Stage 1 & 927 & 610    & 834  & 1043\\
		Stage 2 & 968 & 652    & 850  & 1080 \\
		\hline
	\end{tabular}
\end{table}

\subsubsection{Dynamic Performance}
In this subsection, we analyze the dynamic performance  of the closed-loop system. For comparison, automatic generation control (AGC) is tested in the same scenario. {\color{black} 
In the AGC implementation, the signal of area control error (ACE) is given by $ACE=K_f\omega+P_{ij}$ \cite[Chapter 11.6]{Arthur:Power}, where,  $K_f$ is the frequency response coefficient; $\omega$ the frequency deviation;  $P_{ij}$ the deviation of tie line power. In the case studies, we can treat the whole system as one control area, implying $P_{ij}=0$. Hence, the control center computes $ACE=K_f\omega$ and allocates it to AGC generators, G32, G36, G38 and G39. In this situation, the control command of each generator is $–r_i\cdot ACE$, where $\sum_i r_i=1$. In this paper, we set  $r_i=0.25$ for $i=1,2,3,4$.}
 
The trajectories of frequencies are given in Fig.\ref{frequency}, where the left one stands for the  proposed controller and the right one  for the AGC.
\begin{figure}[!t]
	\centering
	\footnotesize
	\includegraphics[width=0.48\textwidth]{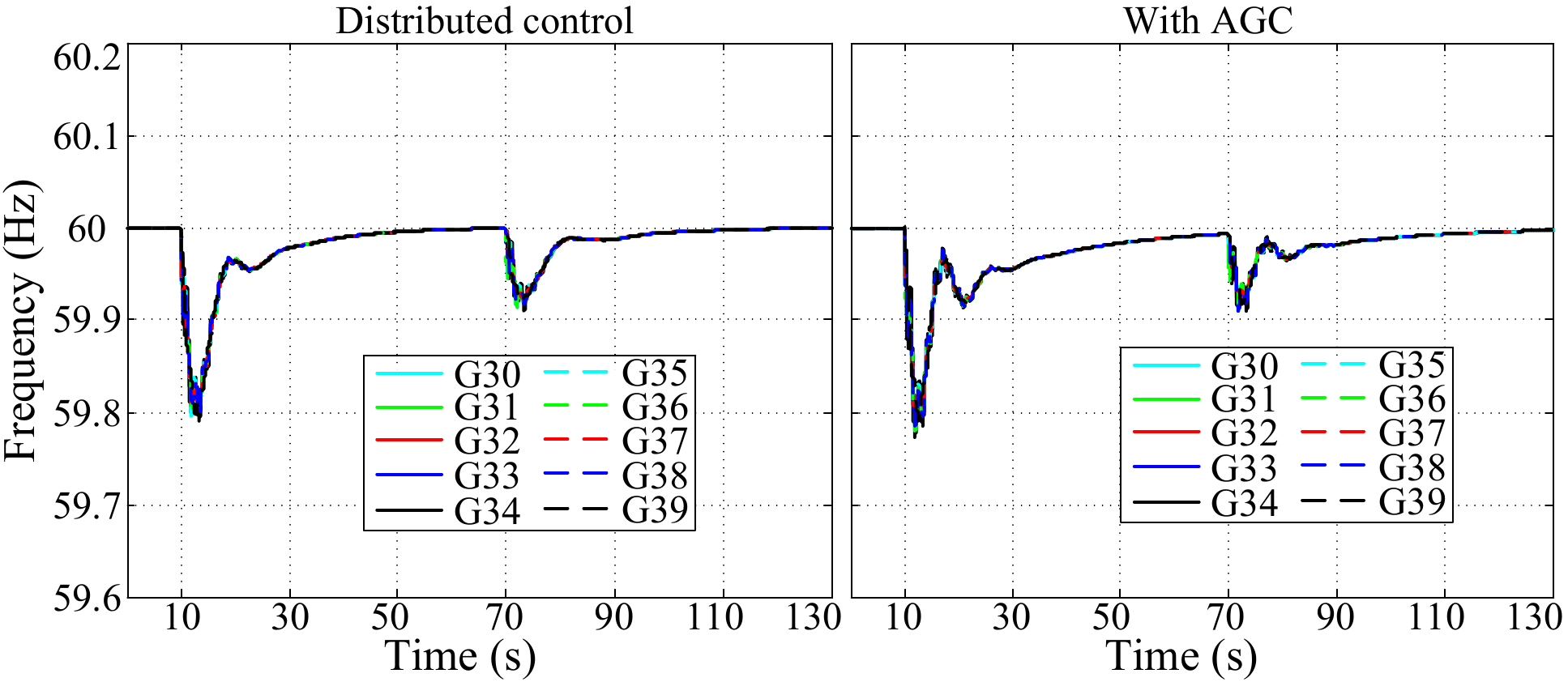}
	\caption{Dynamics of frequencies}
	\label{frequency}
\end{figure}
It is shown in Fig.\ref{frequency} that the frequencies are recovered to the nominal value  under both controls. The frequency drops under two controls are very similar while the recovery time under the proposed control is much less than that under the conventional  AGC. 

{\color{black}
	Although a number of  studies have been devoted to distributed frequency control of power systems, most of them assume that all the nodes are controllable except \cite{Pang:Optimal}. To make a fair comparison, the controller proposed in \cite{Pang:Optimal} is adopted as a rival in our tests. As shown in Eq.(8) of \cite{Pang:Optimal}, each controller needs to predict the load it should supply in steady state. However, it is hard to acquire an accurate prediction in practice, which could lead to steady-state frequency error. }

\begin{figure}[!t]
	\centering
	\footnotesize
	\includegraphics[width=0.38\textwidth]{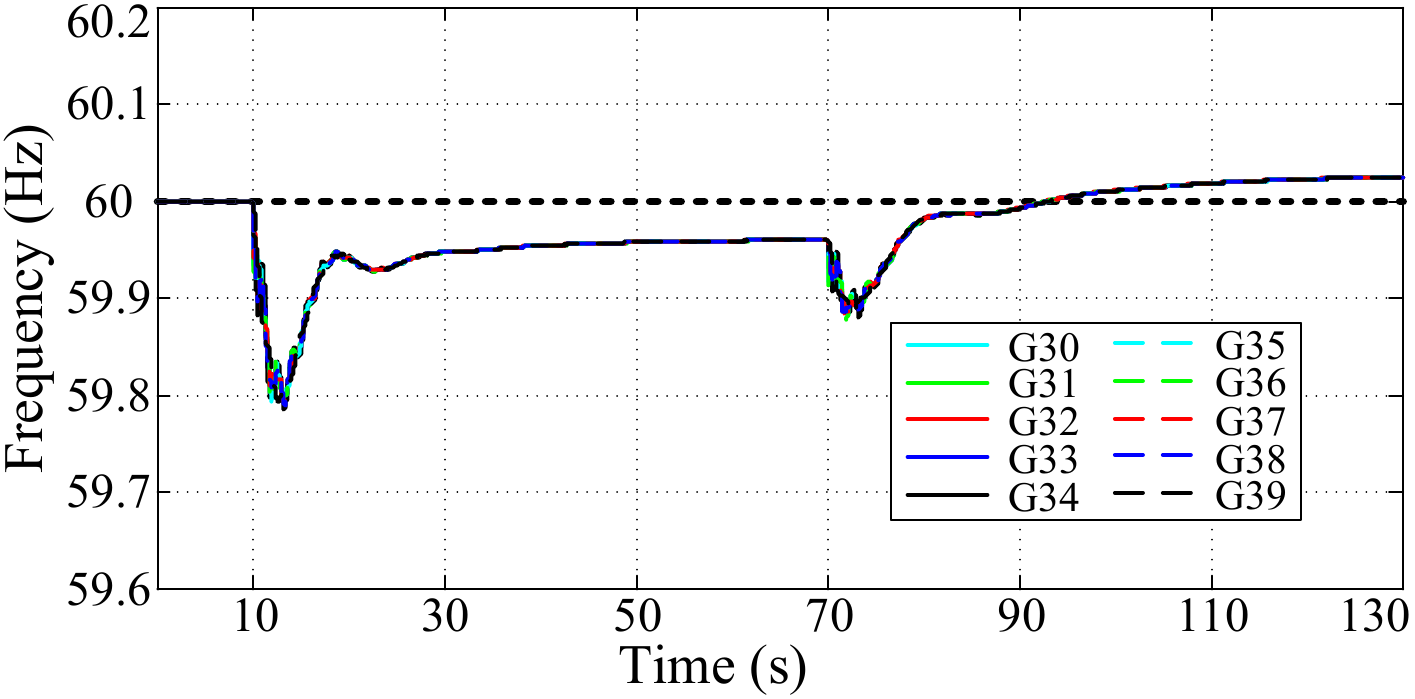}
	\caption{Dynamics of frequencies with the controller in \cite{Pang:Optimal}}
	\label{frequency2}
\end{figure}
\textcolor{black}{
	In the next case, we compare the two controls in the same scenario   as that in Section VII.B.	The dynamics of frequencies with the controller given in \cite{Pang:Optimal} are shown in Fig.\ref{frequency2}. It is observed that there is a frequency deviation in steady state, as the prediction is inaccurate. Although the deviation is usually quite small, it is difficult to completely eliminate. In contrast, when the proposed method is adopted, there is no frequency deviation in steady state, as is shown in Fig.\ref{frequency}. This result is perfectly in coincidence with the indication given by Theorems \ref{thm:optimality}, \ref{thm:stability} and \ref{thm:stability2} and Lemma \ref{Lemma:optimality} in this paper.
}

Mechanical power dynamics under the AGC and the proposed controller  are shown in Fig.\ref{generation_AGC} and Fig.\ref{generation}, respectively. The left parts show  mechanical powers of G32, G36, G38, G39, while the right parts show mechanical powers of other generators adopting conventional controller \eqref{droopcontrol}. With both controls, mechanical powers of the generators adopting \eqref{droopcontrol} remain identical in the steady state. However, there are two  problems when adopting the AGC: 1) mechanical powers are not optimal; 2) mechanical power of G39  violates the capacity limit. In contrast, the proposed control can avoid these problems. In  Stage 1 of Fig.\ref{generation}, no generator reaches capacity limits. In Stage 2, both G38 and G39 reach their upper limits. Then, G38 and G39 stop increasing their outputs while  G32 and G36 raise their outputs  to balance the load demands. In addition, the steady states of  in both stages are optimal, which are the same as shown in Table \ref{tab:eper}.

\begin{figure}[!t]
	\centering
	\includegraphics[width=0.49\textwidth]{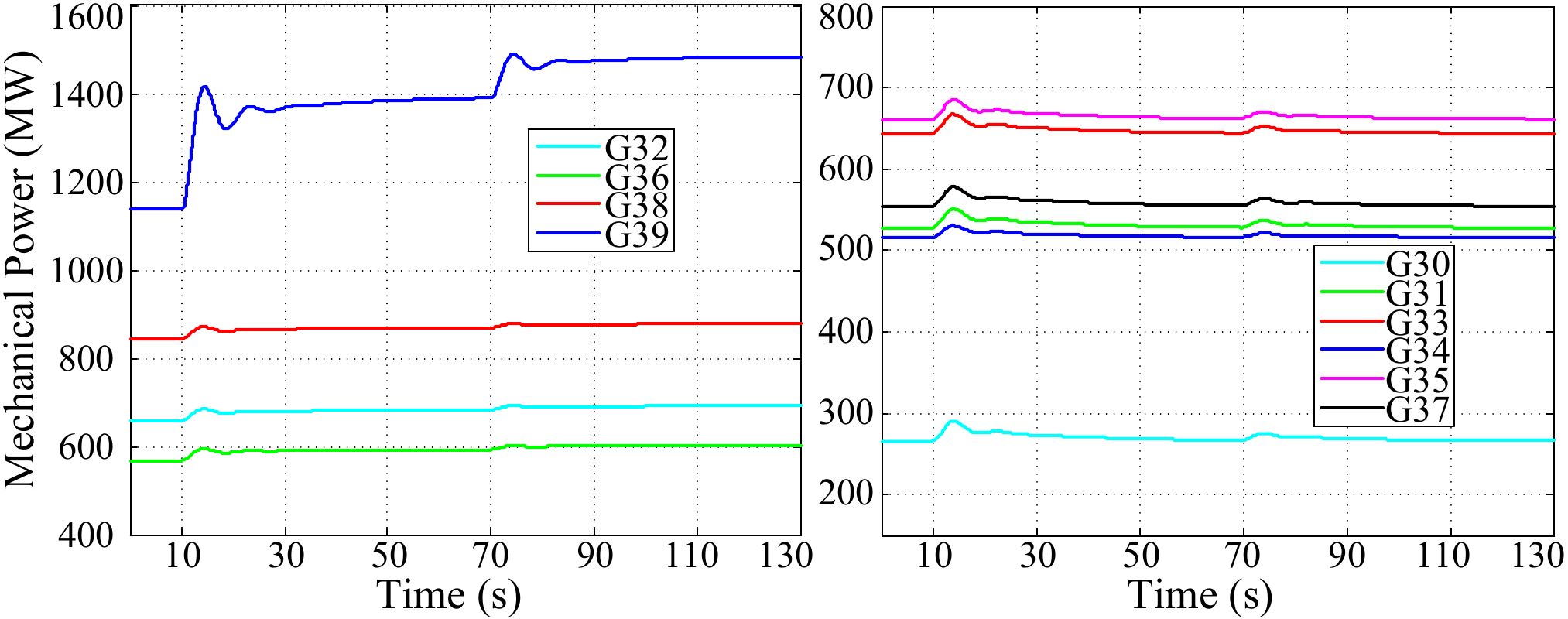}
	\caption{Dynamics of mechanical powers under AGC}
	\label{generation_AGC}
\end{figure}

\begin{figure}[!t]
	\centering
	\includegraphics[width=0.49\textwidth]{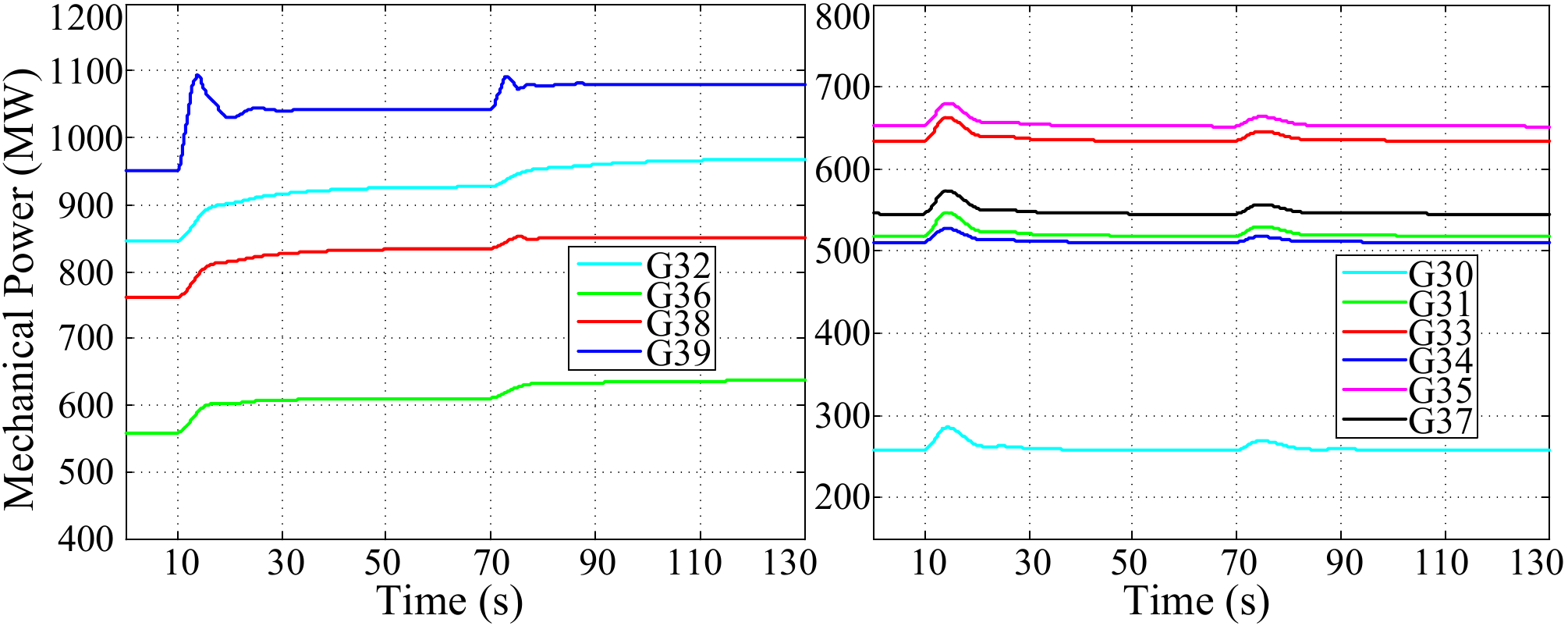}
	\caption{Dynamics of mechanical powers under the proposed control}
	\label{generation}
\end{figure}

\begin{figure}
	\centering
	\includegraphics[width=0.36\textwidth,height=1.3in]{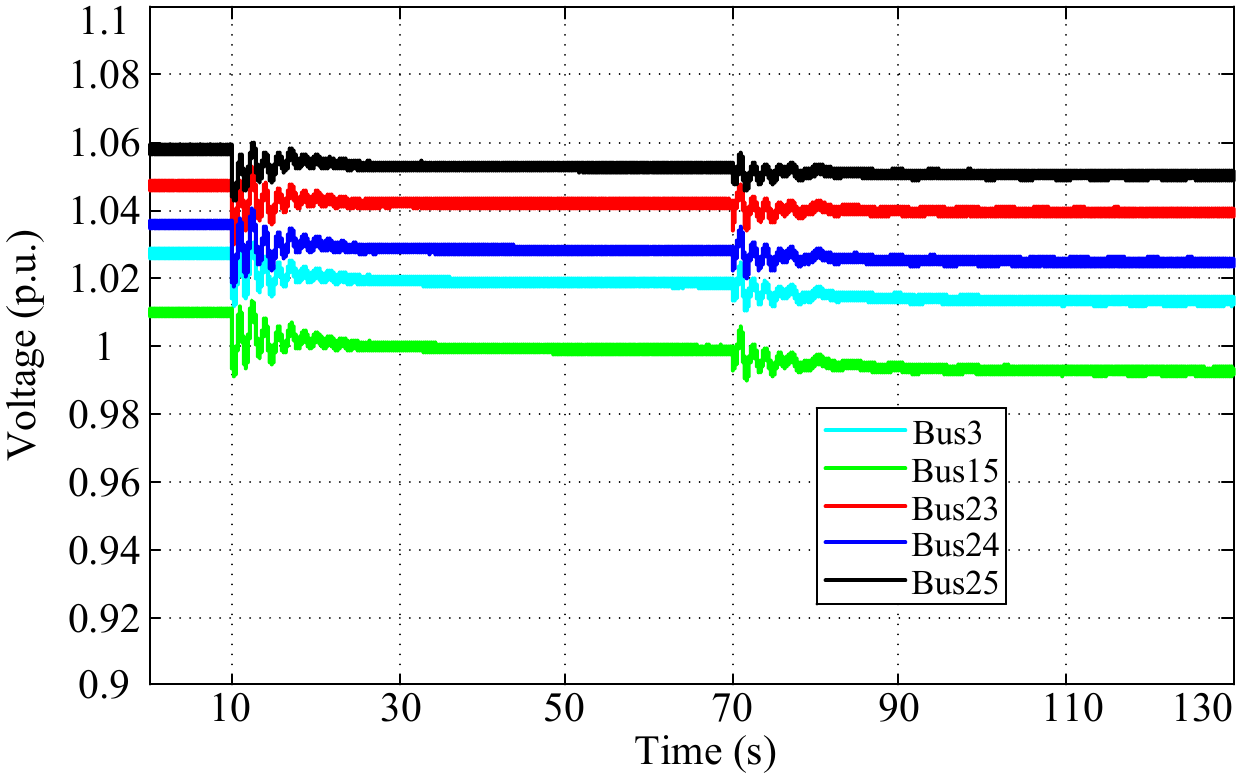}
	\caption{Dynamics of bus voltages}
	\label{networkvoltage}
\end{figure}

We also illustrate in Fig.\ref{networkvoltage} the dynamics of voltage at buses 3, 15, 23, 24, 25. The voltages converge rapidly, and only experience small drops when loads increase. This result validates the effectiveness of the voltage control.

The marginal generation cost of generator $i$, $-\mu_i$ ,  are shown in the left part of Fig.\ref{mu_z}. 
They converge in both stages and the steady-state values 
in  Stage 2 are slightly bigger than that in Stage 1,
as the load changes lead to an increase in the marginal generation cost. Dynamics of $z_{ij}, (i,j)\in E$ are illustrated in the right part of Fig.\ref{mu_z}, which demonstrate that the steady state values do not change in the two stages. In addition, \textcolor{black}{the} variation of $z$ in  transient is very small as the deviation of $\mu_i$ is very small.

%

\begin{figure}[!t]
	\centering
	\includegraphics[width=0.49\textwidth]{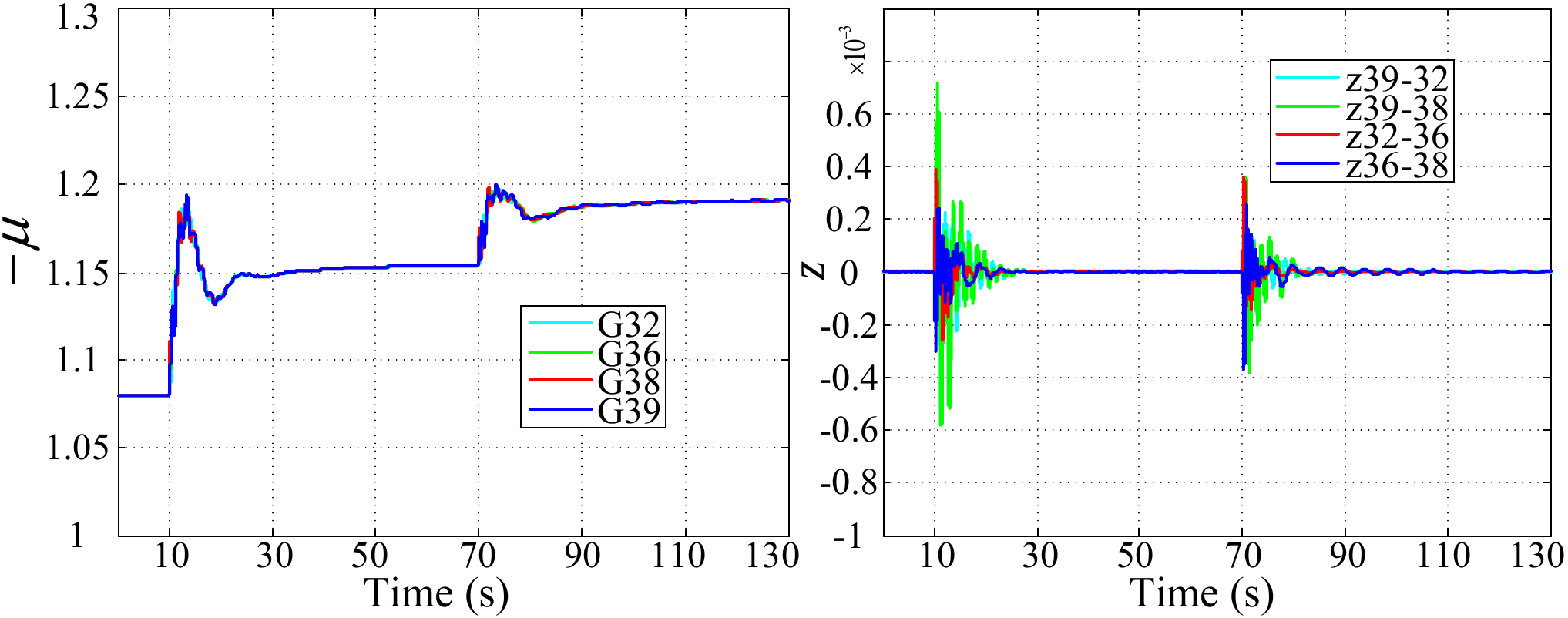}
	\caption{Dynamics of $-\mu$ and $z$}
	\label{mu_z}
\end{figure}

%
%
%

\subsection{Performance under Large Disturbances}
In this subsection, two scenarios of large disturbances are considered. One is  a generator tripping and the other is  a  short-circuit fault followed by a line tripping.
\subsubsection{Generator tripping}
{\color{black} 
At $t=10s$, G32 is  tripped, followed by occurrence of certain power imbalance. Note that the communication graph still remains connected. The output of G32 drops to zero. Frequency and mechanical powers change accordingly. System dynamics are illustrated in Fig.\ref{fre_generation_Gtrip}. The left part of Fig.\ref{fre_generation_Gtrip} shows  the frequency dynamics, and the right shows the mechanical power dynamics of controllable generators. It is observed that the system frequency experiences a big drop at first, and then recovers to the nominal value quickly as other controllable generators increase outputs to balance the power mismatch. These results confirm that our controller can adapt to generator tripping autonomously even if the tripped generator is contributing to  frequency control.	
}
\begin{figure}[t]
	\centering
	\includegraphics[width=0.49\textwidth]{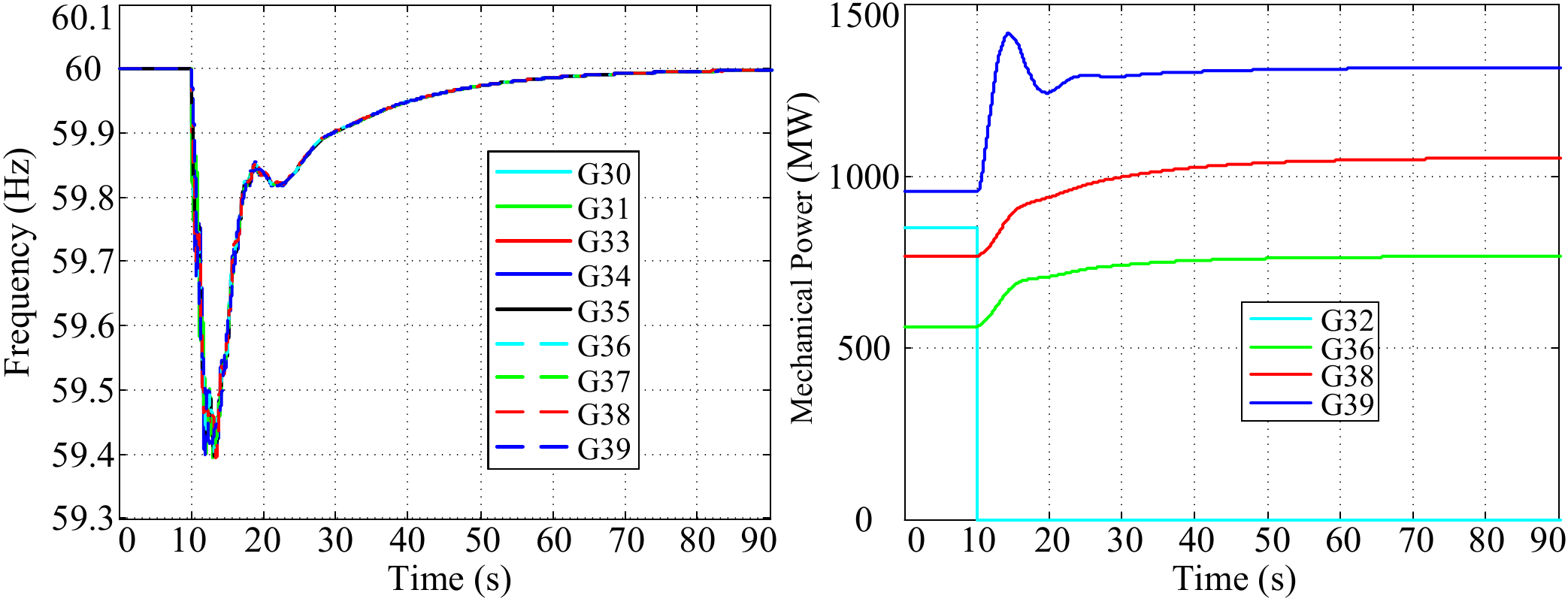}
	\caption{Dynamics of frequencies and mechanical powers with generator trip}
	\label{fre_generation_Gtrip}
\end{figure}

\subsubsection{Short-circuit fault}
At  $t=10s$, there occurs a three-phase short-circuit fault on line $(4,14)$. At $t=10.05$s, this line is tripped by breakers. At $t=60$s, the fault is removed and line $(4,14)$ is re-closed. Frequency dynamics and voltage dynamics of buses 4 and 14  are given in Fig.\ref{fre_voltage_Ltrip}, where the left part shows frequency dynamics and the right shows voltage dynamics. It can be seen that the frequency experiences violent oscillations after the fault happens. And then it is stabilized quickly once the line is tripped. Small frequency oscillation occurs when the line is  re-closed. {\color{black} At the same time, the voltages of buses 4 and 14  drop to nearly zero when the fault happens. The voltages are then stabilized to a new steady-state value in around $10$s after the fault line is tripped. They are slightly different from their initial values because the system’s operating point has changed due to line tripping. When the tripped line is re-closed, the voltages recover to the initial values quickly. } 
\begin{figure}[t]
	\centering
	\includegraphics[width=0.49\textwidth]{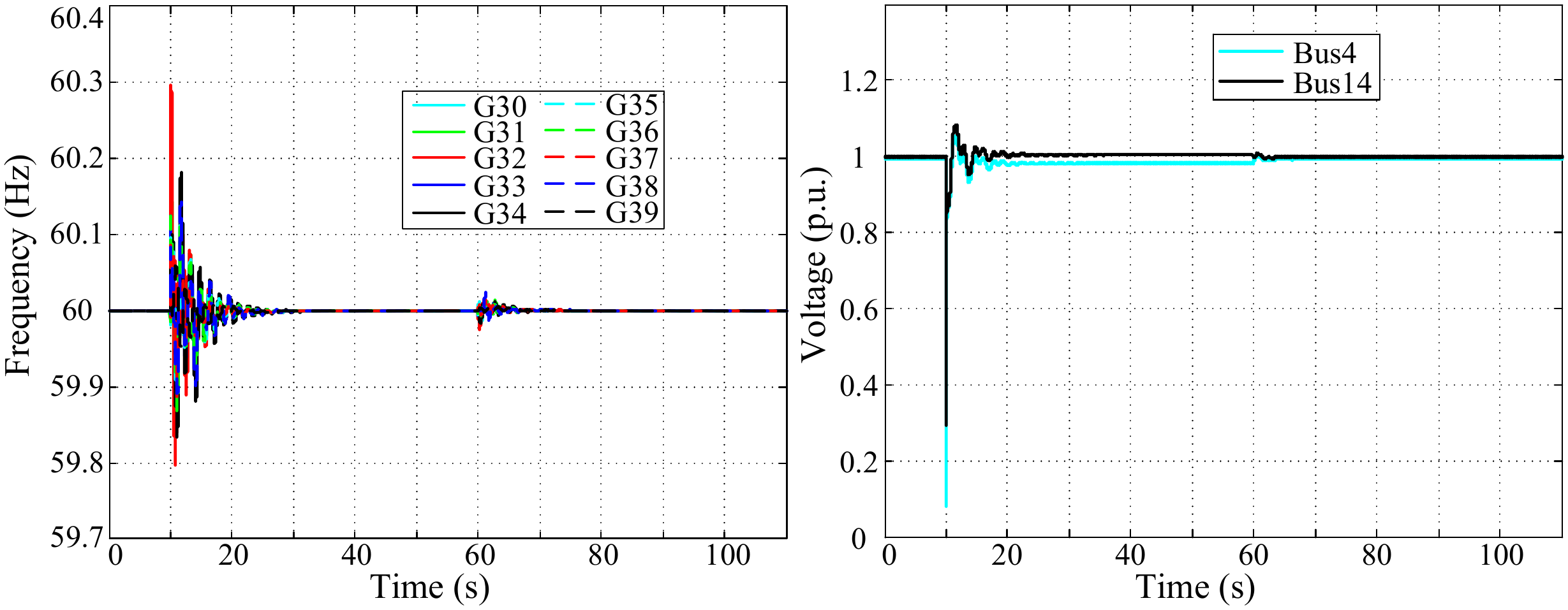}
	\caption{Dynamics of frequencies and voltages with line trip}
	\label{fre_voltage_Ltrip}
\end{figure}

\begin{figure}[t]
	\centering
	\includegraphics[width=0.49\textwidth]{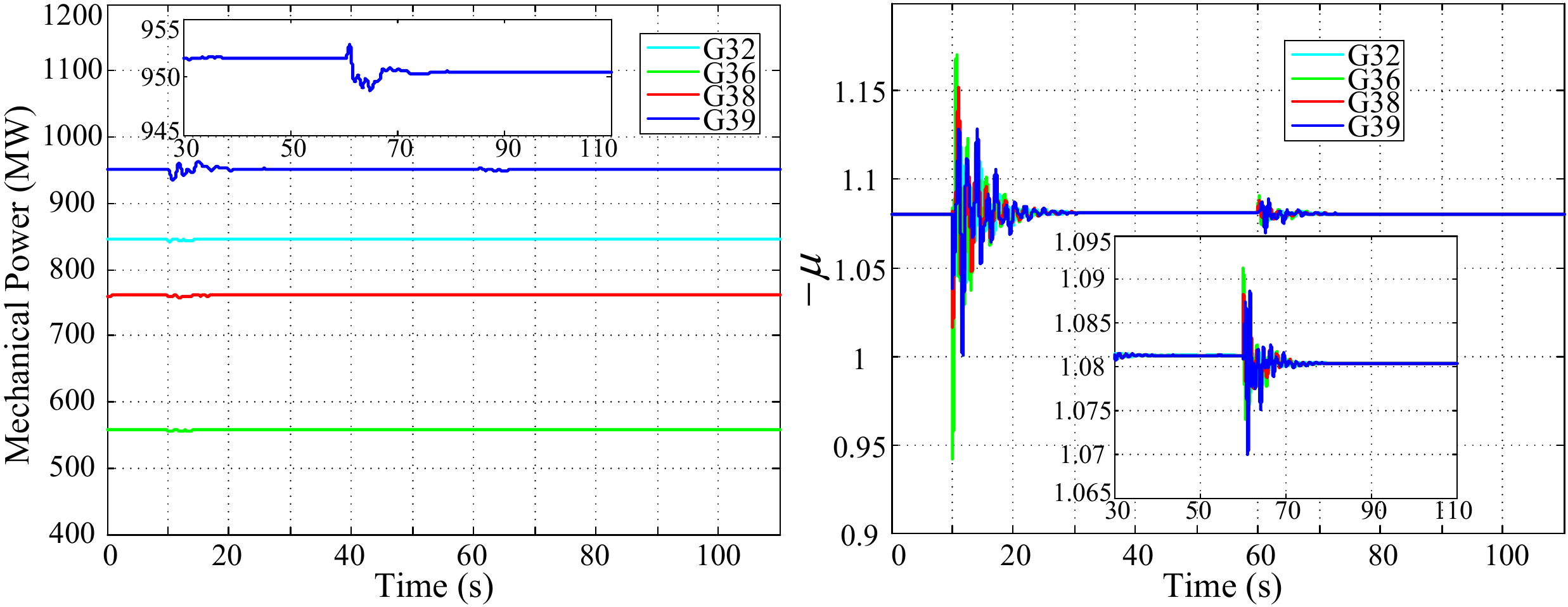}
	\caption{Dynamics of $-\mu$ and mechanical powers with line trip}
	\label{mu_generation_Ltrip}
\end{figure}

When line $(4,14)$ is tripped, the power flow across the power network varies accordingly. As a consequence, the line loss also changes, causing variations of mechanical powers, as shown in Fig.\ref{mu_generation_Ltrip}. In the left part of Fig.\ref{mu_generation_Ltrip}, the inset is the dynamics of mechanical power of G39 from $30$s to $110$s. Similarly, the inset in the right part is dynamics of $-\mu$ of all generators. 
Mechanical powers and their marginal costs all increase when the line is tripped. However, the proposed controller compensates  the loss change autonomously.

These simulation results demonstrate that the proposed distributed optimal frequency controller can cope with large disturbances such as generator tripping and short-circuit fault.

\section{Conclusion}
In this paper, we have designed a distributed optimal frequency control using a nonlinear network-preserving model, where only a subset of generator buses is controlled. We have also simplified the implementation by relaxing the requirements of load measurements and communication topology. Since nonlinearity of  power flow model and dynamics of  excitation voltage has been  taken into account, our controllers can cope with large disturbances. We have proved that the closed-loop system asymptotically converges to the optimal solution of the economic dispatch problem. We have also carried out substantial simulations to verify the good performance of our controller under both small and large disturbances. 

In this work, we have not considered the constraints on line flows, since the controllable generators are selected arbitrarily and may not suffice for congestion management. An interesting problem is how to find out the minimal set of controllable generators to fulfill the requirement of congestion management, which is our future work.



%
%

\bibliographystyle{Bibliography/IEEEtran}
\bibliography{mybib}

\appendix

\makeatletter
\@addtoreset{equation}{section}
\@addtoreset{theorem}{section}
\makeatother
\renewcommand{\theequation}{A.\arabic{equation}}
\renewcommand{\thetheorem}{A.\arabic{theorem}}
\setcounter{equation}{0}

\subsection{Proof of Theorem \ref{thm:optimality}}
\begin{proof}
	From $\dot \gamma _i^-=\dot \gamma _i^+=0$ in \eqref{controller_d}, it follows that $\underline{P}_i^g \le {P}_i^{g*} \le \overline{P}_i^g $, which is the first assertion.
	From $\dot z_{ij}=0$ in \eqref{controller_c}, we get $\mu_i^*=\mu_j^*=\mu_0$. Set $\dot{\mu}_i=0$, add \eqref{controller_b} for $i\in {\cal N}_{CG}$, and recall $\sum\nolimits_{i\in {\cal {N}}_{CG}}\hat p_i = \sum\nolimits_{i\in \cal {N}} p_i- \sum\nolimits_{i\in {\cal {N}}_{UG}} P_i^{g*}$, we have 
	\begin{align}
		\label{eq:equilibrium:pg1}
		\sum\nolimits_{i\in {\cal {N}}_{CG}} P_i^{g*} - \sum\nolimits_{i\in \cal {N}} p_i + \sum\nolimits_{i\in {\cal {N}}_{UG}} P_i^{g*} = 0
	\end{align}
	
	The right sides of (\ref{eq:closedloop.1a1}) and (\ref{eq:closedloop.1a2}) vanish in the equilibrium points, which implies $\omega_i^*=\tilde\omega_i^*=\tilde\omega_j^*=\omega_0$. Set $\dot\omega_i = 0$ and add \eqref{eq:SGmodel.1b}, \eqref{load activepower} and \eqref{load activepower2}. We have 
	\begin{align}
		\label{eq:equilibrium:omega1}
		\omega_0\sum\limits_{i\in {\cal {N}}} \tilde D_i = \sum\limits_{i\in {\cal {N}}_{CG}} P_i^{g*} - \sum\limits_{i\in \cal {N}} p_i + \sum\limits_{i\in {\cal {N}}_{UG}} P_i^{g*} = 0
	\end{align}
	This implies that $\omega_0=0$ due to $\tilde D_i>0$, which is the second assertion. 
	
	Combine \eqref{eq:controlmodel}, \eqref{eq:SGmodel.1d} and $\dot P_i^g=0$. We have 
	$$f_i^{'}(P_i^{g*}) - \gamma _i^{-*} + \gamma _i^{+*} +\omega_i^*+\mu_i^*=0$$
	Since $\omega_0=0$, $\mu_i^*=\mu_j^*=\mu_0$, we can obtain the third assertion.
	
	Next we will prove assertion 4). Since all the constraints of SFC are linear, A3 implies that Slater's condition holds \cite[Chapter 5.2.3]{Boyd:convex}. Moreover, the objective function is strictly convex. We only need to prove that $(P_i^{g*}, \mu_0, \gamma _i^{-*}, \gamma _i^{+*})$ satisfies the KKT condition of SFC in order to prove the fourth assertion. 
		
	The KKT conditions of SFC problem (\ref{SFC}) are 
	\begin{subequations}
		\label{KKT}
		\begin{align}
			\label{KKT1}
			& f_i^{'}(P_i^{g*}) - \gamma _i^{-*} + \gamma _i^{+*} + \mu_0=0 \\
			\label{KKT2}
			& \sum\nolimits_{i\in {\cal {N}}_{CG}} P_i^{g*} - \sum\nolimits_{i\in \cal {N}} p_i + \sum\nolimits_{i\in {\cal {N}}_{UG}} P_i^{g*} = 0 \\ 
			\label{KKT3}
			& \underline{P}_i^g \le {P}_i^{g*} \le \overline{P}_i^g \\
			\label{KKT4}
			& \gamma _i^{-*} \ge 0, \gamma _i^{+*} \ge 0 \\
			\label{KKT5}
			& \gamma _i^{-*}(\underline{P}_i^g - {P}_i^{g*} )=0,  \gamma _i^{+*} ( {P}_i^{g*} - \overline{P}_i^g) =0 
		\end{align}
	\end{subequations}	
	From $\dot \gamma_i^- = \dot \gamma_i^+ = 0$, we have (\ref{KKT3}), (\ref{KKT4}) and (\ref{KKT5}). From the third assertion, we have (\ref{KKT1}). From $\dot{\omega}=0$ and the second assertion, we have (\ref{KKT2}).	Therefore, the equilibrium points of the closed-loop system (\ref{eq:closedloop}) satisfy the KKT conditions (\ref{KKT}). 
	This implies the fourth assertion.
	
	If A3 is strictly satisfied, we know $\exists\ i\in {\cal N}_{CG}$ that $\gamma _i^{-*} = \gamma _i^{+*} =0$. Then, $ \mu_i^*=-f_i^{'}(P_i^{g*}) $ is uniquely determined by $P_i^{g*}$, implying the last assertion. This completes the proof.
\end{proof}

\makeatletter
\@addtoreset{equation}{section}
\@addtoreset{theorem}{section}
\makeatother
\renewcommand{\theequation}{B.\arabic{equation}}
\renewcommand{\thetheorem}{B.\arabic{theorem}}

\subsection{Proof of Theorem \ref{thm:stability}}
\setcounter{equation}{0}
\begin{proof}[Proof of Theorem \ref{thm:stability}]
	Recall \eqref{Lagrangian L}, and dynamics of (\ref{eq:controlmodel}) - (\ref{controller_e}) are rewritten as 
	\begin{subequations}
		\begin{align}
		\label{eq:closedloop.1e2}
		\dot P^g_i & = -{k_{P_i^g}}\cdot\left(\omega_i  + {\partial \hat L(x_1, x_2)}/{\partial P_i^g}\right) \\
		\label{eq:closedloop.1f2}
		\dot \mu_i & = k_{\mu_i}\cdot {\partial \hat L(x_1, x_2)}/{\partial \mu_i}	\\
		\label{eq:closedloop.1g2}
		\dot z_{ij} &= k_{z_i}\cdot {\partial \hat L(x_1, x_2)}/{\partial z_{ij}} \\	
		\label{eq:closedloop.1h2}
		\dot \gamma _i^- & = k_{\gamma_i}\cdot \left[{\partial \hat L(x_1, x_2)}/{\partial \gamma _i^-} \right]_{\gamma _i^-}^+ \\
		\label{eq:closedloop.1i2}
		\dot \gamma _i^+ & = k_{\gamma_i} \cdot\left[{\partial \hat L(x_1, x_2)}/{\partial \gamma _i^+} \right]_{\gamma _i^+}^+ 
		\end{align}
	\end{subequations}
	
	With regard to the closed-loop system, we first define two sets, $\sigma^+$ and $\sigma^-$, as follows \cite{Feijer:Stability}.
	\begin{subequations}
		\begin{align}
		\sigma^+ &:= \{i\in {\cal N}_{CG} \, | \, \gamma^+_{i}=0,         \, P_i^g-\overline P_i^g<0
		\}\\
		\sigma^- &:= \{i\in {\cal N}_{CG} \, | \, \gamma^-_{i}=0,         \, \underline P_i^g-P_i^g<0
		\}
		\end{align}
	\end{subequations}
	Then \eqref{controller_d} and \eqref{controller_e} are equivalent to
	\begin{subequations}
		\label{eq:eta}
		\begin{align}
		\dot\gamma^+_{i} &= \left\{
		\begin{array}{ll}
		k_{\gamma_i}(P_i^g-\overline P_i^g),
		& \text{if}\ i \notin \sigma^+ ;\\
		0,
		& \text{if}\ i \in \sigma^+ .        \end{array}
		\right.\\
		\dot\gamma^-_{i} &= \left\{
		\begin{array}{ll}
		k_{\gamma_i}(\underline P_i^g-P_i^g),
		& \text{if}\ i \notin \sigma^- ;\\
		0,
		& \text{if}\ i \in \sigma^- .        \end{array}
		\right.
		\end{align}
	\end{subequations}
	
	The derivative of $W_k$ is 
	\begin{align}
	\label{derivative of Vk}
	&\dot W_k = \sum\limits_{i \in {\cal N}_{G}}M_i(\omega_i-\omega_i^*)\dot\omega_i + (x-x^*)^T\cdot K^{-1} \dot x \nonumber\\	
	&\le \sum\limits_{i \in {\cal N}_{G}}(\omega_i-\omega_i^*) (P^g_i -D_i \omega_i- P_{ei}) - \sum\limits_{i \in {\cal N}_{CG}}(P_i^g-P_i^{g*})\omega_i \nonumber\\
	&\ \ - (x_1-x^*_1)^T\cdot \nabla_{x_1}\hat L + (x_2-x^*_2)^T\cdot \nabla_{x_2}\hat L  \nonumber\\
	&= \sum\limits_{i \in {\cal N}_{G}}(\omega_i-\omega_i^*) (P^g_i - P^{g*}_i -D_i (\omega_i - \omega_i^*)- (P_{ei}-P_{ei}^*)) \nonumber\\
	&\ \ - \sum\limits_{i \in {\cal N}_{CG}}(P_i^g-P_i^{g*})\cdot(\omega_i- \omega_i^* )  \nonumber\\
	&\ \ - (x_1-x^*_1)^T\cdot \nabla_{x_1}\hat L  + (x_2-x^*_2)^T\cdot \nabla_{x_2}\hat L \nonumber\\
	&= - (x_1-x^*_1)^T\cdot \nabla_{x_1}\hat L  + (x_2-x^*_2)^T\cdot \nabla_{x_2}\hat L \nonumber\\
	&\ \ -\sum\limits_{i \in {\cal N}_{G}}D_i(\omega_i-\omega_i^*)^2 - \sum\limits_{i \in {\cal N}_{G}}(\omega_i-\omega_i^*) (P_{ei}-P_{ei}^*)  \nonumber\\
	&\ \ + \sum\limits_{i \in {\cal N}_{UG}}(P_i^g-P_i^{g*})(\omega_i-\omega_i^*) 
	\end{align}
	where the inequality is due to 
	\begin{align}
	(\gamma^--\gamma^{-*})^T[\underline P^g-P^g]^+_{\gamma^-}&\le(\gamma^--\gamma^{-*})^T(\underline P^g-P^g)\nonumber\\
	&=(\gamma^--\gamma^{-*})^T\nabla_{\gamma^-} \hat L\nonumber
	\end{align}
	Here the inequality holds since $\gamma^-_i=0\le \gamma_i^{-*}$ and $\underline P^g_i-P^g_i < 0$ for $i\in \sigma^-$, i.e., $(\gamma^-_i  -  \gamma_i^{-*}) \cdot (\underline P^g_i-P^g_i)\ge0$. Similarly for $i\in \sigma^+$, we have
	\begin{align}
	(\gamma^+-\gamma^{+*})^T[ P^g-\overline P^g]^+_{\gamma^+}&\le(\gamma^+-\gamma^{+*})^T(P^g-\overline P^g)\nonumber\\
	&=(\gamma^+-\gamma^{+*})^T\nabla_{\gamma^+} \hat L\nonumber
	\end{align}

	From \eqref{load activepower} and \eqref{load activepower2}, it yields 
	\begin{align}
	\label{power balance}
	0&=\sum\limits_{i \in {\cal N}_{G}} (\tilde \omega_i- \tilde\omega_i^*) \bigg((P_{ei}-P_{ei}^*)-\sum\limits_{j \in {N}_{i}} (P_{ij}-P_{ij}^*) \bigg) \\
	&\ \ - \sum\limits_{i \in {\cal N}}\tilde D_i(\tilde\omega_i -\tilde\omega_i^*)^2 - \sum\limits_{i \in {\cal N}_{L}}(\tilde\omega_i -\tilde\omega_i^*)\sum\limits_{j \in {N}_{i}} (P_{ij}-P_{ij}^*) \nonumber
	\end{align}
	Add (\ref{power balance}) to (\ref{derivative of Vk}), then we have
	\begin{align}
	\label{derivative of Vk2}
	&\dot W_k\le -\sum\limits_{i \in {\cal N}_{G}}D_i(\omega_i-\omega_i^*)^2 - \sum\limits_{i \in {\cal N}}\tilde D_i(\tilde\omega_i -\tilde\omega_i^*)^2  \nonumber\\
	& + \sum\limits_{i \in {\cal N}_{G}}(\tilde\omega_i-\omega_i) (P_{ei}-P_{ei}^*) - \sum\limits_{(i,j) \in {\cal E}}(\tilde\omega_i -\tilde\omega_j) (P_{ij}-P_{ij}^*) \nonumber \\
	& + \sum\limits_{i \in {\cal N}_{UG}}(P_i^g-P_i^{g*})(\omega_i-\omega_i^*)   \nonumber\\
	& - (x_1-x^*_1)^T\cdot \nabla_{x_1}\hat L  + (x_2-x^*_2)^T\cdot \nabla_{x_2}\hat L 
	\end{align}
	Since $\hat L$ is a convex function of $x_1$ and a concave function of $x_2$, we have 
	\begin{align}
	\label{saddle point}
	- & (x_1-x^*_1)^T\cdot \nabla_{x_1}\hat L (x_1,x_2) + (x_2-x^*_2)^T\cdot \nabla_{x_2}\hat L (x_1,x_2) \nonumber \\
	\le & \hat L(x_1^*,x_2) - \hat L(x_1,x_2) +\hat L(x_1,x_2) - \hat L(x_1,x^*_2) \nonumber \\
	= &  \hat L(x_1^*,x_2) - \hat L(x_1^*,x_2^*) +\hat L(x_1^*,x_2^*) - \hat L(x_1,x^*_2) \nonumber \\
	\le & \ 0	
	\end{align}
	where the first inequality follows because $\hat L$ is convex in $x_1$ and concave in $x_2$ and the second inequality follows because $(x^*_1, x^*_2)$ is a saddle point. Hence, we have
	\begin{align}
	\label{derivative of Vk3}
	&\dot W_k\le -\sum\limits_{i \in {\cal N}_{G}}D_i(\omega_i-\omega_i^*)^2 - \sum\limits_{i \in {\cal N}}\tilde D_i(\tilde\omega_i -\tilde\omega_i^*)^2  \nonumber\\
	& - \sum\limits_{i \in {\cal N}_{G}}(\omega_i-\tilde\omega_i) (P_{ei}-P_{ei}^*) - \sum\limits_{(i,j) \in {\cal E}}(\tilde\omega_i -\tilde\omega_j) (P_{ij}-P_{ij}^*) \nonumber \\
	& + \sum\nolimits_{i \in {\cal N}_{UG}}(P_i^g-P_i^{g*})(\omega_i-\omega_i^*) 
	\end{align}
	
	The partial of $W_p(x_p)$ is 
	\begin{subequations}
		\begin{align}
		\nabla_{E_{qi}^{'}} W_p&=\frac{1}{x_{di}-x_{di}^{'}}(E_{qi}-E^*_{qi})     \\
		\nabla_{V_{i}} W_p&=0   \\
		\nabla_{\delta_{i}} W_p&=P_{ei}-P^*_{ei} \\ 		
		\nabla_{\theta_{i}} W_p&=\sum\limits_{(i,j) \in {\cal E}} (P_{ij}-P_{ij}^*) - \sum\limits_{i \in {\cal N}_{G}}(P_{ei}-P^*_{ei}) 
		\end{align}
	\end{subequations}
	The derivative of $W_p$ is 
	\begin{align}
	\dot W_p & =  \sum\limits_{i \in {\cal N}_{G}}\frac{(E_{qi}-E^*_{qi})(E_{fi}-E^*_{fi})}{T_{d0i}^{'}(x_{di}-x_{di}^{'})} - \sum\limits_{i \in {\cal N}_{G}}\frac{(E_{qi}-E^*_{qi})^2}{T_{d0i}^{'}(x_{di}-x_{di}^{'})} \nonumber\\
	&+\sum\limits_{i \in {\cal N}_{G}}(\omega_i-\tilde\omega_i) (P_{ei}-P_{ei}^*) + \sum\limits_{(i,j) \in {\cal E}}(\tilde\omega_i -\tilde\omega_j) (P_{ij}-P_{ij}^*)
	\end{align}
	
	The Lyapunov function is defined as $W=W_k+W_p+\sum\limits_{i\in {\cal {N}}_{UG}} S_{\omega_i} + \sum\limits_{i\in {\cal {N}}_{G}} \frac{1}{{T_{d0i}^{'}(x_{di}-x_{di}^{'})}}S_{E_i}$, and its derivative is
	\begin{align}
		\label{dotW}
		\dot W & = \dot W_k+\dot W_p +\sum\nolimits_{i\in {\cal {N}}_{UG}} \dot S_{\omega_i} + \sum\nolimits_{i\in {\cal {N}}_{G}}\frac{\dot S_{E_i}}{{T_{d0i}^{'}(x_{di}-x_{di}^{'})}} \nonumber\\
		&\le-\sum\limits_{i \in {\cal N}_{G}}D_i(\omega_i-\omega_i^*)^2 - \sum\limits_{i \in {\cal N}}\tilde D_i(\tilde\omega_i -\tilde\omega_i^*)^2  \nonumber\\
		& - \sum\limits_{i \in {\cal N}_{G}}\frac{(E_{qi}-E^*_{qi})^2}{T_{d0i}^{'}(x_{di}-x_{di}^{'})}\nonumber\\
		& + \sum\limits_{i \in {\cal N}_{UG}}(P_i^g-P_i^{g*})(\omega_i-\omega_i^*) +\sum\limits_{i\in {\cal {N}}_{UG}} \dot S_{\omega_i} \nonumber\\
		& + \sum\limits_{i \in {\cal N}_{G}}\frac{(E_{qi}-E^*_{qi})(E_{fi}-E^*_{fi})}{T_{d0i}^{'}(x_{di}-x_{di}^{'})} + \sum\limits_{i\in {\cal {N}}_{G}}\frac{\dot S_{E_i}}{{T_{d0i}^{'}(x_{di}-x_{di}^{'})}} \nonumber\\
		&\le-\sum\limits_{i \in {\cal N}_{G}}D_i(\omega_i-\omega_i^*)^2 - \sum\limits_{i \in {\cal N}}\tilde D_i(\tilde\omega_i -\tilde\omega_i^*)^2  \nonumber\\
		& - \sum\limits_{i \in {\cal N}_{G}}\frac{(E_{qi}-E^*_{qi})^2}{T_{d0i}^{'}(x_{di}-x_{di}^{'})} - \sum\limits_{i \in {\cal N}_{UG}} \phi_{\omega_i} - \sum\limits_{i \in {\cal N}_{G}} \phi_{E_i}\nonumber\\
		& \le 0
	\end{align}
	The last two inequalities are due to assumption A4. 
%
	
	To prove the locally asymptotic stability of the closed-loop system, we also need to prove that $W>0, \forall v\in S \backslash v^*$. Equivalently, we need $\nabla^2_v W>0, \forall\ v\in S \backslash v^*$, i.e. A4.
	
	Consequently, there exists a neighborhood set $\{\ v:W(v)\le \epsilon\ \}$ for all sufficiently small $\epsilon>0$ so that $\nabla^2_v W(v)>0$. Hence, there is a compact set $\cal S$ around $v^*$ contained in such neighborhood, which is forward invariant. Let $Z_1:= \{\ v :\dot W(v)=0\ \} $. By LaSalle's invariance principle, the each of trajectories $v(t)$ starting from $\cal S$ converges to the largest invariant set $Z^+$ contained in ${\cal S} \cap Z_1$. From above analysis, if $\dot W(v)=0$, $v$ is an equilibrium point of the closed-loop system \eqref{eq:closedloop}. Hence, $v$ converges to $Z^+ \in \cal V$. 
	
	Finally, we will prove the convergence of each $v(t)$ starting from $\cal V$ is to a point by following the proof of Theorem 1 in \cite{Stegink:aunifying}. Since $v(t)$ is bounded, its $\omega$-limit set $\Omega(v)\neq \emptyset$. By contradiction, suppose there exist two different points in $\Omega(v)$, i.e., $v_1^*, v_2^* \in \Omega(v), v_1^*\neq v_2^*$. Since the Hessian of $W_1(v), W_2(v)$ is positive definite in $\cal S$, there exist $W_1(v), W_2(v)$ defined by \eqref{Lyapunov} with respect to $v_1^*, v_2^*$ and scalars $c_1>0, c_2>0$ such that two sets $W_1^{-1}(\le c_1):=\{v|W_1(v)\le c_1 \}, W_2^{-1}(\le c_2):=\{v|W_2(v)\le c_2 \}$ are disjoint (i.e. $W_1^{-1}(\le c_1)\cap W_2^{-1}(\le c_2)=\emptyset $) and compact. In addition, $W_1^{-1}(\le c_1), W_2^{-1}(\le c_2)$ are forward invariant. By \eqref{dotW}, there exists a finite time $t_1>0$ such that $v(t)\in W_1^{-1}(\le c_1)$ for $\forall t\ge t_1$. Similarly, there exists a finite time $t_2>0$ such that $v(t)\in W_2^{-1}(\le c_2)$ for $\forall t\ge t_2$. This implies that $v(t)\in W_1^{-1}(\le c_1)\cap W_2^{-1}(\le c_2)$ for $\forall t\ge \max (t_1, t_2)$, which is a contradiction.  This completes the proof.
%
%
\end{proof}
{\color{black}
	In this part, we discuss the reasonableness of Assumption A4 by referring to  \cite[Lemma 3]{Persis:A}. Reference \cite{Persis:A} investigates the control of inverter-based  microgrids based on a network-preserving model, while we extend some results to more complicated synchronous-generator based bulk power systems. 
	For simplicity, we first present some notations following \cite{Persis:A}.  Comparing \eqref{line power} and \eqref{eq:Model.PQ}, $P_{ei},Q_{ei}$ have same structures with $P_{ij}, Q_{ij}$, respectively. We can treat the reactance of generator as a line with admittance as $1/{x_{di}^{'}}$, which connects $i\in \mathcal{N}_{G}$ and inner node of the generator. We denote the inner nodes of generators as $\mathcal{N}_{G}^{'}$. Then, we can get a augmented power network, the incidence matrix of which is denoted as $\hat C$. The set of nodes in the augmented power network is denoted as $\mathcal{N}^{'}=\mathcal{N}\cup \mathcal{N}_{G}^{'}$. Denote $\hat V=(E_{q}^{'},V)$, $\hat \theta = (\delta,\theta)$. Let $|\hat C|$ denote the matrix obtained from $\hat C$ by replacing all the elements $c_{ij}$ with $|c_{ij}|$. Define $\Gamma(\hat V):=\text{diag}(|B_{ij}|V_iV_j), i,j \in \mathcal{N}^{'}$. Define $A$ as 
	\begin{align}
	A_{ij}=\left\{
	\begin{array}{ll}
	-|B_{ij}|\cos(\theta_i-\theta_j),
	& i\neq j, i,j\in \mathcal{N} ;\\
	\text{diag}(|B_{ii}|),
	& i= j, i,j\in \mathcal{N};\\
	-{\cos(\delta_i-\theta_j)}/{x_{di}^{'}},
	& i \in \mathcal{N}_{G}^{'}, j \in \mathcal{N}_{G}.        \end{array}
	\right.
	\end{align}	
	For simplicity, we use the following notation. For an $n$-dimensional vector $r:=\{r_1, r_2, \cdots, r_n\}$, the diagonal matrix $\text{diag}(r_1, r_2,\cdots, r_n)$ is denoted in short by $[r]_\mathcal{D}$. And $\textbf{cos}(\cdot),\textbf{sin}(\cdot)$ are defined component-wise.

	From the definition of $W$ in \eqref{Lyapunov}, $\nabla^2_v W(v)>0$ if and only if $\nabla^2_v W_p(v)>0$, i.e. the matrix 
	\begin{align}
		\label{hessian2}
		\left[ {\begin{array}{*{20}{c}} 
			\Gamma(\hat{V})[\textbf{cos}(\hat C^T {\hat \theta})]_\mathcal{D}  & 
			[\textbf{sin}(\hat C^T {\hat \theta})]_\mathcal{D}\Gamma(\hat{V})|\hat C|^T[\hat V]_\mathcal{D}^{-1}\\
			{[\hat V]_\mathcal{D}^{-1}}|\hat C|\Gamma(\hat{V})[\textbf{sin}(\hat C^T {\hat \theta})]_\mathcal{D}& 
			A+H(\hat V)
			\end{array}} \right]
	\end{align}
	is positive definite, where 
	\begin{align}
		H(\hat V)=\left[ \begin{array}{*{20}{c}} 
		\left[\frac{x_{di}}{2x^{'}_{di}(x_{di}-x_{di}^{'})}\right]_\mathcal{D}  & 0\\
		0 & [{q_i}/{V_i^{2}}]_\mathcal{D}
		\end{array} \right]
	\end{align}

	In any steady state of power system (i.e.), the phase-angle difference between two neighboring nodes is usually small. In addition, the difference between $\delta_i$ and $\theta_i$ is also small. This implies that the matrix in \eqref{hessian2} is diagonal dominant as well as its positive definiteness.
	Therefore, Assumption A4 is usually satisfied and makes sense.
}

\makeatletter
\@addtoreset{equation}{section}
\@addtoreset{theorem}{section}
\makeatother
\renewcommand{\theequation}{C.\arabic{equation}}
\renewcommand{\thetheorem}{C.\arabic{theorem}}
\subsection{Proof of Lemma \ref{Lemma:optimality}}
\begin{proof}
	From $\dot z_{ij}=0$, we get $\mu_i^*=\mu_j^*=\mu_0$. Set $\dot{\mu}_i=0$ and add \eqref{load estimation} for  $i\in {\cal N}_{CG}$, we have 
	\setcounter{equation}{0}
	\begin{align}
	\label{eq:equilibrium:pg2}
	\sum\nolimits_{i\in {\cal {N}}_{CG}}M_i\dot\omega_i+\sum\nolimits_{i\in {\cal {N}}_{CG}}(D_i+\tau_i)\omega_i = 0
	\end{align}
	
	The right sides of (\ref{eq:closedloop.1a1}) and (\ref{eq:closedloop.1a2}) vanish in the equilibrium points, which implies $\omega_i^*=\tilde\omega_i^*=\tilde\omega_j^*=\omega_0$. Set $\dot\omega=0$ in \eqref{eq:equilibrium:pg2} and we have 
	\begin{align}
	\label{eq:equilibrium:omega2}
	\omega_0\sum\nolimits_{i\in {\cal {N}}_{CG}} (D_i+\tau_i) =0
	\end{align}
	This implies that $\omega_0=0$ due to $D_i+\tau_i > 0$. 
	
	Other proofs are same with that in Theorem \ref{thm:optimality}, which are omitted. 
\end{proof}

\makeatletter
\@addtoreset{equation}{section}
\@addtoreset{theorem}{section}
\makeatother
\renewcommand{\theequation}{D.\arabic{equation}}
\renewcommand{\thetheorem}{D.\arabic{theorem}}
\subsection{Proof of Theorem \ref{thm:stability2}}
\begin{proof}[Proof of Theorem \ref{thm:stability2}]
\setcounter{equation}{0}

We still use the Lyapunov function \eqref{Lyapunov} to analyze the stability of the closed-loop system \eqref{eq:closedloop2}. 
Denote $y=(\omega_i^T,x_1^T,x_2^T)^T, i\in {\cal N}_{CG}$, and define the following function
\begin{equation}
	\tilde f\left( y \right) = \left[ {\begin{array}{*{20}{c}}
		 	-D_i \omega_i\\
			-\left(  f_i^{'}(P_i^g)+\mu_i-\gamma _i^- + \gamma _i^+\right)\\
			f_{\mu_i}     \\
			\mu_i-\mu_j\\
			\underline{P}_i^g - {P}_i^g  \\
			{P}_i^g - \overline{P}_i^g   
			\end{array}} \right], i\in {\cal N}_{CG}
\end{equation}
where $f_{\mu_i}=P_i^g - \hat p_i - \sum\limits_{j \in {N_{ci}}} {\left( {{{ \mu }_i} - {{\mu }_j}} \right)}  - \sum\limits_{j \in {N_{ci}}} {{z_{ij}}} -\tau_i\mu_i - \tau_if_i^{'}(P_i^g) +\beta_i\omega_i$\footnote{Sometimes, we also use $\beta_i$ instead of $\beta_i(t)$ for simplification. In addition, if  \eqref{eq:A3} is not binding, we can omit $\gamma _i^-, \gamma _i^+$ in  neighborhoods of equilibrium points}.

The derivative of $W_k$ is 
\begin{align}
\label{derivative of Vk4}
&\dot W_k = \sum\limits_{i \in {\cal N}_{G}}M_i(\omega_i-\omega_i^*)\dot\omega_i + (x-x^*)^T\cdot K^{-1} \dot x \nonumber\\	
&= \sum\limits_{i \in {\cal N}_{G}}(\omega_i-\omega_i^*) (P^g_i - P^{g*}_i -D_i (\omega_i - \omega_i^*)- (P_{ei}-P_{ei}^*)) \nonumber\\
&  + (x-x^*)^T\cdot K^{-1}\dot x \nonumber\\
& \le \sum\limits_{i \in {\cal N}_{G}}(\omega_i-\omega_i^*) (P^g_i - P^{g*}_i - (P_{ei}-P_{ei}^*)) \nonumber\\
& -\sum\limits_{i \in {\cal N}_{UG}} D_i (\omega_i - \omega_i^*)^2  + (y-y^*)^T\tilde f(y)
\end{align}
where the inequality is due to the same reason for that in \eqref{derivative of Vk}.

Divide  $\dot W_k$ into two parts, and $\dot W^1_k=(y-y^*)^T \tilde f(y)$, $\dot W^2_k=\dot W_k-\dot W^1_k$. Then we will analyze the sign of $\dot W^1_k$.
\begin{align}
	&\dot W^1_k = (y-y^*)^T\tilde f(y) \nonumber\\
	& =  \int_0^1 (y-y^*)^T \frac{\partial }{\partial \tilde z}\tilde f(\tilde z(s)) (y-y^*) ds + (y-y^*)^T \tilde f(y^*)\nonumber\\
	& \le  \frac{1}{2}\int_0^1 (y-y^*)^T \left[\frac{\partial^T }{\partial \tilde z} \tilde f(\tilde z(s))  + \frac{\partial }{\partial\tilde z} \tilde f(\tilde z(s)) \right] (y-y^*) ds \nonumber\\
	\label{dot Wk1}
	& = \int_0^1 (y-y^*)^T \left[H(\tilde z(s)) \right] (y-y^*) ds
\end{align}
where $\tilde z(s)=y^*+s(y-y^*)$. The second equation is from the fact that $\tilde f(y)-\tilde f(y^*) = \int_0^1 \frac{\partial }{\partial\tilde z} \tilde f(\tilde z(s)) (y-y^*) ds$. The inequality is due to either $\tilde f(y^*) =0$ or $\tilde f(y^*) < 0, y_i\ge 0$, i.e. $(y-y^*)^T \tilde f(y^*)\le0$.
\begin{align}
	 \frac{\partial\tilde f(y) }{\partial y}  = -\left[ {\begin{array}{*{20}{c}}
	 	D&       0&                0&  0& 0&  0\\
	 	0&        \nabla^2_{P^g} f(P^g)& I&  0& -I& I\\
	 	\beta&   \tau\nabla^2_{P^g} f(P^g)-I&   \tau+L_c& C& 0&  0\\
	 	0&        0&             -C^T&  0& 0&  0\\
	 	0&       I&                0&  0& 0&  0\\
	 	0&        -I&                0&  0& 0&  0
	 	\end{array}} \right]
\end{align}
where $D=\text{diag}(D_i)$, $\beta=\text{diag}(\beta_i)$, $\tau=\text{diag}(\tau_i)$, $I$ is the identity matrix with dimension $n_{CG}$, $C$ is the incidence matrix of the communication graph, $L_c$ is the Laplacian matrix of the communication graph.

Finally, $H$ in \eqref{dot Wk1} is
\begin{align}
	H&=\frac{1}{2}\left[\frac{\partial^T }{\partial y} \tilde f(y)  + \frac{\partial }{\partial y} \tilde f(y) \right] \nonumber\\
     &=	\left[ {\begin{array}{*{20}{c}}
     	-D&       0& -\frac{1}{2}\beta&  0&  0&  0\\
     	0& -\nabla^2_{P^g} f(P^g)&  -\frac{\tau}{2}\nabla^2_{P^g} f(P^g)&  0&  0&  0\\
     	-\frac{1}{2}\beta&    -\frac{\tau}{2}\nabla^2_{P^g} f(P^g)&  -\tau-L_c&  0&  0&  0\\
     	0&        0&                0&  0&  0&  0\\
     	0&        0&                0&  0&  0&  0\\
     	0&        0&                0&  0&  0&  0
     	\end{array}} \right]
\end{align}
$H\le0$, if
\begin{align}
\left[ {\begin{array}{*{20}{c}}
	-D&       0& -\frac{1}{2}\beta\\
	0& -\nabla^2_{P^g} f(P^g)&  -\frac{\tau}{2}\nabla^2_{P^g} f(P^g)\\
	-\frac{1}{2}\beta&    -\frac{\tau}{2}\nabla^2_{P^g} f(P^g)&  -\tau
	\end{array}} \right]<0
\end{align}
By Schur complement \cite{lin:distributed}, we only need 
\begin{align}
\label{Schur complememt}
	-\tau_i - \left[ {\begin{array}{*{20}{c}}
		-\frac{1}{2}\beta_i&    -\frac{\tau_i}{2}c_i
		\end{array}} \right] 
	  \left[ {\begin{array}{*{20}{c}}
		-D_i&       0\\
		0& -c_i
		\end{array}} \right]^{-1}
	  \left[ {\begin{array}{*{20}{c}}
		-\frac{1}{2}\beta_i\\   -\frac{\tau_i}{2}c_i
		\end{array}} \right] <0
\end{align}
where $c_i=\nabla^2_{P^g_i} f(P^g_i)$. Solving \eqref{Schur complememt}, we can get
\begin{align}
	-\sqrt{\tau_i D_i(4-\tau_i c_i)}<\beta_i<\sqrt{\tau_i D_i(4-\tau_i c_i)}
\end{align}
By A2, we know $c_i\le l_i$, thus 
$$\sqrt{\tau_i D_i(4-\tau_i l_i)}\le\sqrt{\tau_i D_i(4-\tau_i c_i)},$$ $$-\sqrt{\tau_i D_i(4-\tau_i l_i)}\ge-\sqrt{\tau_i D_i(4-\tau_i c_i)}.$$
Here we need $\tau_i>0,\ 4-\tau_i l_i>0$, i.e. $0<\tau_i<4/l_i$.
Finally, we have 
\begin{align}
	-\sqrt{\tau_i D_i(4-\tau_i l_i)}<\beta_i<\sqrt{\tau_i D_i(4-\tau_i l_i)}
\end{align}
i.e. $\bar\beta_i<\sqrt{\tau_i D_i(4-\tau_i l_i)}$, implying \eqref{beta_range}.

Analysis of $\dot W_k^2$ and $\dot W_p$ as well as the convergence to a point are same as those in the proof of Theorem \ref{thm:stability}, which are omitted here.
\end{proof}




%
%
%
%
%
%
%

\end{document}